\documentclass[12pt,draft]{article}

\usepackage{mathrsfs,bm,color,sasaki,braket}
\usepackage{amsmath,amssymb,amsthm}
\usepackage{comment}

\usepackage[a4paper]{geometry}
\newgeometry{left=30mm,right=30mm,bottom=32mm,top=32mm}

\theoremstyle{plain}
\newtheorem{theorem}{Theorem}[section]

\newtheorem{thm}[theorem]{\bf Theorem}

\newtheorem{lem}[theorem]{\bf Lemma}

\newtheorem{prop}[theorem]{\bf Proposition}

\numberwithin{equation}{section}

\theoremstyle{remark}
\newtheorem{rem}[theorem]{\bf Remark}

\title{\sc Explicit Diagonalization of Pair Interaction Models}
\author{Yasumichi Matsuzawa\thanks{Department of Mathematics, Faculty of Education, Shinshu University,
        6-Ro, Nishi-nagano, Nagano 380-8544, Japan,
        e-mail: myasu@shinshu-u.ac.jp},
\and
        Itaru Sasaki\thanks{Department of Mathematics, Shinshu University, Matsumoto 390-8621, Japan, 
        e-mail: isasaki@shinshu-u.ac.jp},
\and 
        Kyosuke Usami\thanks{Department of Mathematics, Shinshu University, Matsumoto 390-8621, Japan,
        e-mail: 19SS101F@shinshu-u.ac.jp}
}
\date{\today}

\begin{document}
\maketitle       

\begin{abstract}
We provide a general method for constructing bosonic Bogoliubov transformations that diagonalize
 a general class of quadratic Hamiltonians.
These Hamiltonians describe the pair interaction models.
Bogoliubov transformations are constructed algebraically, and the resulting Hamiltonians 
become the second quantizations of explicit one-particle Hamiltonians.
Moreover, an explicit formula for the ground state energies is given. 
Our method systematically diagonalizes various models of quantum field theory,
 including a model of a harmonic oscillator coupled to a Bose field and the Pauli-Fierz models
 in the dipole approximation.
\end{abstract}

\tableofcontents 

\section{Introduction}\label{intro}
The Bogoliubov transformation is a basic tool for analyzing quantum field Hamiltonians.
It is a map from the creation and annihilation operators 
to their linear combinations, which preserves the canonical commutation relations (CCRs).
It is believed that Hamiltonians with quadratic interactions of creation and annihilation
operators can be diagonalized by Bogoliubov transformations \cite{Be}.
We remark, however, that the construction of the Bogoliubov transformations is a non-trivial mathematical problem.
Here the diagonalization means that the Hamiltonian is unitarily equivalent to the second quantization of a one-particle Hamiltonian, 
up to a constant.

In many studies of quadratic Hamiltonians, Bogoliubov transformations were constructed using the scattering theory, 
and extra regularities on coupling functions were required.
Readers are referred to \cite{KM} for physical discussions and \cite{Ar81, A83, A83b, AF, EB} for mathematical studies.
The paper \cite{GS} concerns a case that the spectrum of a one-particle Hamiltonian is purely discrete.
However, the systematic construction of Bogoliubov transformations that diagonalize quadratic 
Hamiltonians has not been fully elucidated.
In \cite{De17,TNS}, it has been shown that a wide class of quadratic Hamiltonians are diagonalized, 
but the resulting one-particle Hamiltonians are not specific enough for further analysis.

In this paper, we focus on Hamiltonians of the form
\begin{align*}
  H = \dGb(T) + \frac{1}{2}\sum_{n=1}^\infty \lambda_n \PhiS(g_n)^2,
\end{align*}
where $\dGb(T)$ is the second quantization operator of $T$, 
$\lambda_n\in\RR$ and $\PhiS(g_n)$ is the Segal field operator with a coupling function $g_n$.
We call the quantum field model described by $H$ the pair interaction model.
Under the conditions (B1)--(B6) for $T$, $\lambda_n$ and $g_n$ given in Section \ref{DefModel},
we explicitly diagonalize these Hamiltonians, meaning that the resulting one-particle Hamiltonian
 and the ground state energy are explicitly given.
Many physical models satisfy those conditions, and 
we will analyze concrete physical examples in Section \ref{examples}.
We remark that our construction of Bogoliubov transformations does not use any scattering theory, is algebraic, 
and is independent of precise spectral properties of one-particle Hamiltonians.

Let us state our main results, Theorem \ref{main} and Theorem \ref{gse}, precisely.
Under the conditions (B1)--(B6), the Hamiltonian $H$ is self-adjoint and bounded from below.
Moreover, $H$ is explicitly diagonalized 
by a unitary operator $U$ implementing a Bogoliubov transformation, that is, 
\begin{align*}
  UHU^* = \dGb(S) + E,
\end{align*}
where the one-particle Hamiltonian $S$ and the ground state energy $E$ are given by
\[
S = \Big(T^2+\sum_{n=1}^\infty \lambda_n \ket{T^{1/2}g_n}\!\bra{T^{1/2}g_n}\Big)^{1/2},\ \ \ \ \ 
E = \frac{1}{2} \mathrm{tr}(\overline{S-T}).
\]
Thus, after the diagonalization, the spectral analysis of $H$ is reduced to 
that of $S$.
Notably, if $T$ is given by the relativistic dispersion relation $T=(-\Delta+m^2)^{1/2}$, 
then $S^2$ becomes a trace class perturbation of the free Schr\"odinger operator.


This paper is organized as follows.
In Section \ref{generalD}, we give a criterion for diagonalizing a Hamiltonian by a Bogoliubov 
transformation in a general setting (Theorem \ref{diagonalization} and Theorem \ref{diag prop}).
In Section \ref{construct}, we construct a class of Bogoliubov transformations from two non-negative self-adjoint operators (Theorem \ref{SP2}).
In Section \ref{DefModel}, we define the Hamiltonians of the pair interaction models and prove the self-adjointness (Theorem \ref{saH}).
In Section \ref{Diagonalization}, we diagonalize the Hamiltonians (Theorem \ref{main}). 
In Section \ref{Sec:gse}, we give an explicit formula for ground state energies (Theorem \ref{gse}).
In Section \ref{examples}, we apply those results to various models of quantum field theory. 
More precisely, we consider the single pair interaction model, a model of a harmonic oscillator coupled to a Bose field, and the Pauli-Fierz models in the dipole approximation. 
In Appendix A, we give some inequalities on the creation-annihilation operators and second quantization operators.
In Appendix B, we show the equality of the domains of $T^{-3/2}$ and $S^{-3/2}$ under a suitable condition, 
which is used to solve the infrared problem of the Pauli-Fierz model (Theorem \ref{e/a gs}).

\section{General Theory of Diagonalization}\label{generalD}

In this section, we give a criterion for diagonalizing a Hamiltonian by a Bogoliubov 
transformation in a general setting.
Let us begin with recalling a boson Fock space and related objects.

Let $\sH$ be a separable complex Hilbert space. The boson Fock space over $\sH$ is defined by
\begin{align*}
  \Fb(\sH) := \bigoplus_{n=0}^\infty \bigg[ \sTensor^n \sH \bigg],
\end{align*}
where $\otimes_{\rm s}^n \sH$ denotes the $n$-fold symmetric tensor product of $\sH$ and $\otimes_{\rm s}^0 \sH:=\CC$.
A vector $\Psi \in \Fb(\sH)$ is denoted by $\Psi = (\Psi^{(n)})_{n=0}^\infty$ with 
$\Psi^{(n)}\in \stensor^n \sH$.
The standard creation operator $A^*(f)$ for $f\in\sH$ is defined by
\begin{align*}
& \dom(A^*(f)) := \left\{\Psi=(\Psi^{(n)})_{n=0}^\infty \in \Fb(\sH) ~\bigg|~ \sum_{n=0}^\infty n\norm{S_n(f\tensor\Psi^{(n-1)})}^2<\infty\right\} \\
& (A^*(f)\Psi)^{(n)} := \sqrt{n}S_n(f\tensor \Psi^{(n-1)}), \qquad n=1,2,3,\cdots,
\end{align*}
and $(A^*(f)\Psi)^{(0)}:=0$.
Here, $S_n$ is the symmetrization operator, which is an orthogonal projection from $\otimes^n \sH$ onto $\otimes_{\rm s}^n \sH$.
The adjoint operator $A(f) := [A^*(f)]^*$ is called the annihilation operator.
Let
\[
\Fbz := \{ \Psi=(\Psi^{(n)})_{n=0}^\infty \in \Fb(\sH) \mid \Psi^{(n)} = 0, n\geq N \text{ for some } N\in\NN \}.
\]
For a subspace $\sD \subset \sH$, we set 
\begin{align*}
 \Ffin(\sD) := \mathrm{L.h.}\{\Omega, A^*(f_1)\cdots A^*(f_n)\Omega \mid n\in\NN, f_j\in\sD, j=1,\cdots,n \},
\end{align*}
where $\Omega:=(1,0,0,\cdots)\in\Fb(\sH)$ is called the Fock vacuum.
It follows that $A(f)\Omega = 0 ~ (f\in\sH)$.
Note that $\Fbz$ and $\Ffin(\sD)$ are cores for $A(f)$ and $A^*(f)$ if $\sD$ is dense in $\sH$.
The operators $A(f), A^*(f)$ satisfy the following CCRs on $\Fbz$:
\begin{align*}
  & [A(f),A^*(g)] = \inner{f}{g} \\
  & [A(f),A(g)] = 0 = [A^*(f), A^*(g)]  ~~~~~~   f,g\in\sH.
\end{align*}
For $f\in\sH$, the Segal field operator is defined by
\begin{align*}
  \PhiS(f) = \frac{1}{\sqrt{2}} \overline{(A(f)+A^*(f))}.
\end{align*}
They satisfy the CCRs in the following form:
\[
 e^{i\PhiS(f)}e^{i\PhiS(h)} = e^{-i \Im\langle f,h\rangle /2}e^{i\PhiS(f+h)},
 \qquad f,h\in\sH.
\]
It is known that the set $\{e^{i\PhiS(f)} \mid f\in\sH\}$ is irreducible, 
meaning that the only everywhere defined bounded operators on $\Fb(\sH)$ 
that commute with all $e^{i\PhiS(f)} ~ (f\in\sH)$ are scalar operators. 

We denote by $\mathcal{B}(\sH)$ the set of everywhere defined bounded operators on $\sH$.
Let $J$ be a conjugate operator on $\sH$.
Suppose that two bounded operators $X,Y\in\mathcal{B}(\sH)$ satisfy
\begin{align}
 X^*X - Y^*Y = 1,   \qquad X^*JYJ - Y^*JXJ = 0,   \label{simp1}
\intertext{and}
 XX^* - JYY^*J = 1, \qquad - XY^* + JYX^*J = 0.   \label{simp2}
\end{align}
For each $f\in\sH$, we define an operator $B(f)$ by
\begin{align}\label{B(f)}
 & B(f) := \overline{A(Xf)+A^*(JYf)}.
\end{align}
Here, $\overline{S}$ denotes the closure of a closable operator $S$. 
By \eqref{simp1}, the operators $B(f),B^*(f)$ $(f\in\sH)$ satisfy the CCRs on $\Fbz$, i.e.,
\begin{align}
  [B(f), B^*(g)] = \inner{f}{g}, \qquad [B(f),B(g)]=0,
  \qquad f,g \in \sH.
\end{align}
The correspondence
\[
  \{A(f),A^*(f)\mid f\in\sH\}\mapsto \{B(f),B^*(f)\mid f\in\sH\}
\]
is called a Bogoliubov transformation.
It is known that there exists a unitary operator $U$ on $\Fb(\sH)$ such that
\begin{align}
  U B(f) U^* = A(f), \qquad f\in\sH   \label{BT}
\end{align}
if and only if $Y$ is Hilbert-Schmidt (see \cite{Rui78}). 
In this case, the Bogoliubov transformation is said to be proper.
We call the unitary operator $U$ a (proper) Bogoliubov transformation too.
Set 
\begin{align}
 \mathfrak{Sp}
& := \{ (X,Y) \in \cB(\sH)\times\cB(\sH) 
        \mid \eqref{simp1} \text{ and } \eqref{simp2} \text{ hold} \}  \label{sp} \\
 \mathfrak{Sp}_2 
& := \{ (X,Y) \in \mathfrak{Sp} \mid Y \text{ is Hilbert-Schmidt} \}.   \label{sp2}
\end{align}

Let $(X,Y)\in\mathfrak{Sp}$ be arbitrary, and let $B(f)$ be as defined in \eqref{B(f)}.
We set
\begin{align*}
  \phi(f) := \frac{1}{\sqrt{2}} \overline{(B(f)+B^*(f))}, \qquad f\in\sH,
\end{align*}
which is the field operator corresponding to $B(f)$.
Then $\{\phi(f)\mid f\in\sH\}$ satisfies the CCRs:
\[
 e^{i\phi(f)}e^{i\phi(h)} = e^{-i \Im\langle f,h\rangle /2}e^{i\phi(f+h)},
 \qquad f,h\in\sH.
\]
We set 
\begin{equation}\label{defF}
 F(f) := Xf+JYf, \qquad f\in\sH.
\end{equation}
Then $\phi(f)=\PhiS(Xf+JYf)=\PhiS(F(f))$.

\begin{lem}\label{bijective}
The map $\sH\ni f\mapsto F(f)=Xf+JYf\in\sH$ is bijective, continuous and real-linear.
\end{lem}

\begin{proof}
It is sufficient to show that $F$ is bijective.
For each $f\in\sH$, we put $G(f):=X^*f-Y^*Jf$.
Then
\begin{align*}
(G\circ F)(f) &= X^*F(f)-Y^*JF(f) = X^*(Xf+JYf)-Y^*J(Xf+JYf)\\
&= (X^*X-Y^*Y)f + (X^*JYJ-Y^*JXJ)Jf=f,
\end{align*}
and
\begin{align*}
(F\circ G)(f) &= XG(f)+JYG(f) = X(X^*f-Y^*Jf)+JY(X^*f-Y^*Jf)\\
&= (XX^*-JYY^*J)f+(-XY^*+JYX^*J)Jf = f.
\end{align*}
Thus $F$ is bijective.
\end{proof}

\begin{lem}\label{irreducible}
Let $\sD $ be a dense subspace of $\sH$.
Then $\{e^{i\phi(f)}\mid f\in\sD \}$ is irreducible.
\end{lem}

\begin{proof}
By Lemma \ref{bijective}, the set $\{Xf+JYf \mid f\in\sD \}$ is dense in $\sH$.
Hence
\[
\{e^{i\phi(f)}\mid f \in \sD \} = \{e^{i\Phi_{\rm S}(Xf+JYf)}\mid f\in \sD \}
\]
is irreducible.
\end{proof}

In what follows, we assume that $Y$ is Hilbert-Schmidt, i.e., $(X,Y)\in\mathfrak{Sp}_2$.
Thus there exists a unitary operator $U$ on $\Fb(\sH)$ such that
\[
 UB(f)U^* = A(f), \qquad f\in\sH.
\] 

\begin{lem}
It follows that
\[
 U\phi(f)U^* = \PhiS(f), \qquad f\in\sH.
\]
\end{lem}

\begin{proof}
We first note that $\phi(f)$ is essentially self-adjoint on $\Ffin(\sH)$ because $\phi(f)=\PhiS(Xf+JYf)$.
Since
\[
 UB(f)U^*=A(f), \qquad UB^*(f)U^*=A^*(f),
\]
we obtain
\[
 U \left( \frac{B(f)+B^*(f)}{\sqrt{2}}\right) U^* = \frac{A(f)+A^*(f)}{\sqrt{2}},
\]
whence
\[
 U\phi(f)|_{\Ffin(\sH)}U^* 
 = U\left(\frac{B(f)+B^*(f)}{\sqrt{2}}\right)\Big|_{\Ffin(\sH)}U^*
   \subset \frac{A(f)+A^*(f)}{\sqrt{2}} \subset \PhiS(f).
\]
By taking the closure of both sides, we get the desired result.
\end{proof}

Before going to the next theorem, we recall the second quantization operator.
Let $S$ be a self-adjoint operator acting in $\sH$.
The second quantization $\dGb(S)$ of $S$ is a self-adjoint operator acting in $\Fb(\sH)$, which is defined by
\begin{align*}
  \dGb(S) := \bigoplus_{n=0}^\infty \overline{S^{(n)}},
\end{align*}
where 
\begin{equation}\label{S(n)}
S^{(n)}:= {\sum_{j=1}^{n} \one\tensor\cdots \one\tensor
\stackrel{j\mathrm{th}}{S}\tensor \one\cdots \tensor \one}
\end{equation}
with
$S^{(0)}:=0$.

\begin{theorem}\label{diagonalization}
Let $H$ be a self-adjoint operator acting in $\Fb(\sH)$, and let $S$ be a self-adjoint operator acting in $\sH$.
Suppose that there exists a dense subspace $\sD$ of $\sH$ such that
\begin{equation}\label{covariance}
 e^{itH}\phi(f)e^{-itH}=\phi(e^{itS}f), \qquad  t\in \RR,\ f\in\sD.
\end{equation}
Then there exists a real number $E$ so that $UHU^*=\dGb(S)+E$. 
\end{theorem}

\begin{proof}
Recall that $\Omega=(1,0,0,\cdots)\in\Fb(\sH)$ is the Fock vacuum.
Fix $t\in\RR$.
We first show that two vectors $e^{-itH}U^*\Omega$ and $U^*\Omega$ are linearly dependent.
Let $\mathcal{A}$ be the linear span of $\{e^{i\phi_(f)} \mid f\in\sD\}$.
It follows that for any $f\in\sD$,
\begin{align*}
 &\langle e^{-itH}U^*\Omega,e^{i\phi(f)}e^{-itH}U^*\Omega\rangle 
  = \langle U^*\Omega,e^{itH}e^{i\phi(f)}e^{-itH}U^*\Omega\rangle\\ 
   &= \langle U^*\Omega,e^{i\phi(e^{itS}f)}U^*\Omega\rangle
= \langle \Omega,Ue^{i\phi(e^{itS}f)}U^*\Omega\rangle = \langle \Omega,e^{i\PhiS(e^{itS}f)}\Omega\rangle \\
&= e^{-\|f\|^2/4}= \langle \Omega,e^{i\PhiS(f)}\Omega\rangle = \langle \Omega,Ue^{i\phi(f)}U^*\Omega\rangle =  \langle U^*\Omega,e^{i\phi(f)}U^*\Omega\rangle.
\end{align*}
By linearity,
\begin{equation}\label{aiueo}
\langle e^{-itH}U^*\Omega,xe^{-itH}U^*\Omega\rangle = \langle U^*\Omega,xU^*\Omega\rangle
\end{equation}
holds for any $x\in \mathcal{A}$.
Thanks to the CCRs, $\mathcal{A}$ is a $*$-algebra containing the identity.
Since $\{e^{i\phi(f)} \mid f\in \sD\}$ is irreducible by Lemma \ref{irreducible}, 
the von Neumann's double commutant theorem implies that
$\mathcal{A}$ is weakly dense in $\mathcal{B}(\Fb(\sH))$.
Thus the equality \eqref{aiueo} holds for any $x\in\mathcal{B}(\Fb(\sH))$.
Letting $x=|U^*\Omega\rangle\langle U^*\Omega|$, we get $|\langle e^{-itH}U^*\Omega,U^*\Omega\rangle|=1$.
Since the Cauchy-Schwarz inequality becomes an equality if and only if two vectors are linearly dependent,
we conclude that $e^{-itH}U^*\Omega$ and $U^*\Omega$ are linearly dependent.

Take a unique complex number $c(t)$ so that $e^{-itH}U^*\Omega=c(t)U^*\Omega$.
Then $|c(t)|=1$.
Since the map $t\mapsto e^{-itH}$ is a continuous group homomorphism, so is $t\mapsto c(t)$, 
and hence there exists a unique real number $E$ such that 
\[
c(t)=e^{-itE}, \qquad  t\in\RR.
\]
We next show that $UHU^*= \dGb(S)+E$.
For this, let $t\in\RR$ and $f\in\sD$ be arbitrary.
Then
\begin{align*}
&Ue^{itH}e^{i\phi(f)}U^*\Omega 
 = Ue^{itH}e^{i\phi(f)}e^{-itH}\cdot e^{itH}U^*\Omega
  = Ue^{i\phi(e^{itS}f)}\cdot e^{itE}U^*\Omega\\
&=e^{itE}e^{i\PhiS(e^{itS}f)}\Omega = e^{itE}e^{it\dGb(S)}e^{i\PhiS(f)}\Omega
  = e^{it(\dGb(S)+E)}e^{i\PhiS(f)}\Omega.
\end{align*}
On the other hand, 
\[
Ue^{itH}e^{i\phi(f)}U^*\Omega = Ue^{itH}U^*\cdot Ue^{i\phi(f)}U^*\Omega = e^{itUHU^*}e^{i\PhiS(f)}\Omega.
\]
Since the linear span of $\{e^{i\PhiS(f)}\Omega \mid f\in\sD\}$ is dense in $\Fb(\sH)$,
we obtain $UHU^*= \dGb(S)+E$.
\end{proof}

The ground state energy $E$ can be expressed in terms of $H$, $S$ and $Y$.
\begin{prop}\label{gse1}
Assume the conditions of Theorem \ref{diagonalization}.
If $S$ is non-negative and $\Omega\in\dom(H)$, then $\overline{YS^{1/2}}$ is Hilbert-Schmidt and
the ground state energy $E$ of $H$ is given by 
 \begin{align*}
    E = \inner{\Omega}{H\Omega} - \norm{\overline{YS^{1/2}}}_\mathrm{HS}^2.
 \end{align*}
\end{prop}

\begin{proof}
By Theorem \ref{diagonalization}, we have $H=U^*\dGb(S)U+E$.
Since $\Omega\in\dom(H)$, one has $U\Omega \in\dom(\dGb(S))$ and 
\begin{align*}
 E = \inner{\Omega}{H\Omega} - \inner{U\Omega}{\dGb(S)U\Omega}.
\end{align*}
Since $S$ is non-negative, for any orthonormal basis $\{e_n\}_n\subset \dom(S)$, we have (see e.g., \cite[Theorem 5.21]{A18})
\begin{align}
 \inner{U\Omega}{\dGb(S)U\Omega} 
 = \sum_{n=1}^\infty \norm{A(S^{1/2}e_n)U\Omega}^2. \label{conv001}
\end{align}
By the definition of $U$, we get
\begin{align*}
 \norm{A(S^{1/2}e_n)U\Omega}
= \norm{B(S^{1/2}e_n)\Omega}
 = \norm{A^*(JYS^{1/2}e_n)\Omega}
 = \norm{JYS^{1/2}e_n}
= \norm{YS^{1/2}e_n}.
\end{align*}
Since the right-hand side of \eqref{conv001} converges, so does $\sum_{n}\norm{YS^{1/2}e_n}^2$.
Thus, $\overline{YS^{1/2}}$ is Hilbert-Schmidt and 
\begin{align*}
  \inner{U\Omega}{\dGb(S)U\Omega} = \norm{\overline{YS^{1/2}}}_\mathrm{HS}^2.
\end{align*}
This finishes the proof.
\end{proof}

The following theorem provides a sufficient condition for the assumption \eqref{covariance} of Theorem \ref{diagonalization}.
For densely defined closed operators $A,B$, we define a quadratic form
\begin{align*}
  \inner{\Phi}{[A,B]_\mathrm{w}\Psi}
 := \inner{A^*\Phi}{B\Psi}-\inner{B^*\Phi}{A\Psi}, 
\end{align*}
for $\Phi\in\dom(A^*)\cap \dom(B^*)$ and $\Psi\in\dom(A)\cap\dom(B)$.

\begin{theorem}\label{diag prop}
Let $H$ be a self-adjoint operator acting in $\Fb(\sH)$, and let $S$ be an injective non-negative self-adjoint
operator acting in $\sH$. 
Assume the following conditions:
\begin{enumerate}
\item[(i)] There exist a dense subspace $\sD_1\subset \sH$ such that $\Ffin(\sD_1)\subset \dom(H)$.
\item[(ii)] $\dom(H) \subset \dom(\dGb(S)^{1/2})$ holds.
\item[(iii)] There exist a dense subspace $\sD\subset \dom(S)$ so that 
$e^{itS}\sD \subset \sD\ (t\in\RR)$ and that $F(f),F(Sf)\in\dom(S^{-1/2})$ for all $f\in \sD$, where $F$ is defined in \eqref{defF}.
\item[(iv)] For all $f\in\sD$, 
  \begin{align*}
   \lim_{\vep\to 0}\norm{S^{-1/2}F\Big(\Big(\frac{e^{i\vep S}-1}{\vep} -iS\Big)f\Big)} =0.    
  \end{align*}
\item[(v)] For all $f\in\sD$ and $ \Psi,\Phi\in\dom(H)$,
\begin{align}
& \inner{\Phi}{[H,B(f)]_\mathrm{w}\Psi} = -\inner{\Phi}{B(Sf)\Psi},   \label{comm1}\\ 
& \inner{\Phi}{[H,B^*(f)]_\mathrm{w}\Psi} = \inner{\Phi}{B^*(Sf)\Psi}. \label{comm2}
\end{align}
\end{enumerate}
Then
\begin{align}
  e^{itH} \phi(f) e^{-itH} = \phi(e^{itS}f), \quad 
  t\in\RR, f\in\sD. \label{transH}
\end{align}
In particular, $UHU^* = \dGb(S)+E$ for some $E\in\RR$.
\end{theorem}

\begin{proof}
Let $f\in\sD$.
Since $F(f), F(if)\in\dom(S^{-1/2})$, we have
\begin{equation}\label{dominc1}
 \dom(H)\subset\dom(\dGb(S)^{1/2}) \subset \dom(B(f)) \cap \dom(B^*(f)).
\end{equation}
Similarly, it follows from $F(Sf)\in\dom(S^{-1/2})$ that
\begin{equation}\label{dominc2}
 \dom(H)\subset\dom(\dGb(S)^{1/2}) \subset \dom(B(Sf)) \cap \dom(B^*(Sf)).
\end{equation}
Hence \eqref{comm1} and \eqref{comm2} are well-defined for any $\Psi,\Phi\in\dom(H)$.
Let $\Psi,\Phi\in\dom(H)$ and set 
$f_t:=e^{-itS}f$ and $\Psi_t:=e^{-itH}\Psi,~\Phi_t:=e^{-itH}\Phi$.
Then $f_t\in\sD$ and $\Psi_t, \Phi_t\in \dom(H)$.
We show that the function
\begin{align*}
  X(t) := \inner{\Phi_t}{\phi(f_t)\Psi_t}
\end{align*}
is differentiable in $t$ and
\begin{align}
 X'(t) 
 = i\inner{\Phi_t}{[H,\phi(f_t)]_\mathrm{w}\Psi_t}
   -\inner{\Phi_t}{\phi(iSf_t)\Psi_t}.   \label{eq:X'}
\end{align}
We set
\begin{align*}
\Delta_\vep\Psi := \vep^{-1}(\Psi_{t+\vep}-\Psi_t), \qquad
\Delta_\vep\Phi := \vep^{-1}(\Phi_{t+\vep}-\Phi_t), \qquad 
\Delta_\vep f := \vep^{-1}(f_{t+\vep}-f_t).
\end{align*}
Then
\begin{align*}
& \frac{X(t+\vep)-X(t)}{\vep} \\
& = \inner{\Delta_\vep\Phi}{\phi(f_{t+\vep})\Psi_{t+\vep}}
 + \inner{\phi(\Delta_\vep f)\Phi_t}{\Psi_{t+\vep}}
 + \inner{\phi(f_t)\Phi_t}{\Delta_\vep\Psi}.
\end{align*}
Since $\Psi_t\in\dom(H)$, 
we have that $\Delta_\vep\Psi\to -iH\Psi_t \, (\vep\to 0)$, strongly. Thus
\begin{align}
\lim_{\vep\to 0} \inner{\phi(f_t)\Phi_t}{\Delta_\vep\Psi}
 =\inner{\phi(f_t)\Phi_t}{-iH\Psi_t}.   \label{eqqxx1}
\end{align}
We use the following standard estimate
\begin{align}
& \norm{\phi(g)\Xi} \leq 2\norm{S^{-1/2}F(g)}  \norm{\dGb(S)^{1/2}\Xi}
  + \norm{F(g)} \norm{\Xi},
  \quad g\in\sD, ~ \Xi \in \dom(\dGb(S)^{1/2}). \label{stdineq}
\end{align}
By condition (ii) and the closed graph theorem, there exist constants $C_1,C_2>0$ such that 
\begin{align}
  \norm{\dGb(S)^{1/2}\Xi} \leq C_1\norm{H\Xi}+C_2\norm{\Xi}, \qquad \Xi\in\dom(H).
 \label{boundSH}
\end{align}
By using \eqref{stdineq} and \eqref{boundSH}, we have
\begin{align*}
&\norm{\phi(f_{t+\vep})\Psi_{t+\vep} - \phi(f_t)\Psi_t} \\
&\leq \norm{\phi(f_{t+\vep}-f_t) \Psi_{t+\vep}} + \norm{\phi(f_t) (\Psi_{t+\vep}-\Psi_t)}\\
&\leq C\norm{S^{-1/2}F(f_{t+\vep}-f_t)} (\norm{H\Psi_{t+\vep}} + \norm{\Psi_{t+\vep}})
  + C\norm{F(f_{t+\vep}-f_t)} \norm{\Psi_{t+\vep}} \\
&\quad  + C\norm{S^{-1/2}F(f_t)} (\norm{H(\Psi_{t+\vep}-\Psi_t)} + \norm{\Psi_{t+\vep}-\Psi_t})
  + C\norm{F(f_t)} \norm{\Psi_{t+\vep}-\Psi_t}  \\
&= C\norm{S^{-1/2}F(f_{t+\vep}-f_t)} (\norm{H\Psi} + \norm{\Psi})
  + C\norm{F(f_{t+\vep}-f_t)} \norm{\Psi} \\
&\quad  + C\norm{S^{-1/2}F(f_t)} (\norm{H(\Psi_\vep-\Psi)} + \norm{\Psi_\vep-\Psi})
  + C\norm{F(f_t)} \norm{\Psi_\vep-\Psi}
\end{align*}
for some $C>0$.
By condition (iv), $\norm{S^{-1/2}F(f_{t+\vep}-f_t)}$ goes to zero as $\vep\to 0$.
Thus we have $\norm{\phi(f_{t+\vep})\Psi_{t+\vep} - \phi(f_t)\Psi_t}\to 0\ (\vep \to 0)$, and hence
\begin{align}
  \lim_{\vep\to 0}  \inner{\Delta_\vep\Phi}{\phi(f_{t+\vep})\Psi_{t+\vep}}
 = \inner{-iH\Phi_t}{\phi(f_t)\Psi_t}. \label{eqqxx2}
\end{align}
By condition (iv), we have $S^{-1/2}F(\Delta_\vep f+iSf_t)\to 0\ (\vep\to 0)$.
By noting this fact, we can show that 
\begin{align}
&\lim_{\vep\to 0} \inner{\phi(\Delta_\vep f)\Phi_t}{\Psi_{t+\vep}}
 = -\inner{\phi(iSf_t)\Phi_t}{\Psi_t}
 = -\inner{\Phi_t}{\phi(iSf_t)\Psi_t}. \label{eqqxx3}
\end{align}
Therefore, by combining \eqref{eqqxx1}, \eqref{eqqxx2} and \eqref{eqqxx3}, 
we have that $X(t)$ is differentiable in $t$ and \eqref{eq:X'} holds.
By \eqref{comm1} and \eqref{comm2}, we get 
\begin{align*}
X'(t) 
&= \frac{i}{\sqrt{2}}\inner{\Phi_t}{[H,B(f_t)+B^*(f_t)]_\mathrm{w}\Psi_t}
  -\frac{1}{\sqrt{2}}\inner{\Phi_t}{(B(iSf_t)+B^*(iSf_t))\Psi_t} \\
&= \frac{i}{\sqrt{2}}\inner{\Phi_t}{(B(-Sf_t) + B^*(Sf_t))\Psi_t}
  -\frac{1}{\sqrt{2}}\inner{\Phi_t}{(B(iSf_t)+B^*(iSf_t))\Psi_t} \\
&= 0.
\end{align*}
Thus, $X(t)=X(0)$ for all $t$ and hence
\begin{align*}
  \inner{e^{-itH}\Phi}{\phi(e^{-itS}f)e^{-itH}\Psi}
  = \inner{\Phi}{\phi(f)\Psi}
\end{align*}
for all $f\in\sD$, $\Psi,\Phi\in\dom(H)$, which implies
\begin{align}
  e^{itH} \phi(f) e^{-itH} \Big|_{\dom(H)} = \phi(e^{itS}f) \Big|_{\dom(H)}. \label{transdomH}
\end{align}
Since $\sD_1$ is dense, $\phi(f)$ is essentially self-adjoint on $\Ffin(\sD_1)$ (see e.g., \cite[Theorem 5.22]{A18}).
Thus, by $\Ffin(\sD_1)\subset \dom(H)$, $\dom(H)$ is a core for $\phi(f)~(f\in\sD)$.
Therefore, by taking the closure of \eqref{transdomH}, we get \eqref{transH}.
The relation \eqref{transH} and Theorem \ref{diagonalization} imply that 
$UHU^*=\dGb(S)+E$ for some $E\in\RR$.
\end{proof}

\section{Construction of Proper Bogoliubov Transformations}\label{construct}
In this section, we construct proper Bogoliubov transformations
that will be used to diagonalize the Hamiltonians of the pair interaction models in Section \ref{Diagonalization}.
For this, we construct a pair of operators $(X,Y) \in \mathfrak{Sp}_2$ from two non-negative self-adjoint operators $S$ and $T$.

For self-adjoint operators $A, B$, we write $A \leq B$ if $\dom(B) \subset \dom(A)$ and $\inner{f}{Af} \leq \inner{f}{Bf}$ holds for all $f \in \dom(B)$.
Let us introduce conditions for $S$ and $T$ as follows:
\begin{enumerate}
\item[(A1)] $S, T$ are injective non-negative self-adjoint operators acting in $\sH$. 
\item[(A2)] There are positive constants $c_1>0, c_2>0$ such that $c_1^2 S^2\leq T^2 \leq c_2^2S^2$. 
\item[(A3)] $\left(\overline{ST^{-1}}\right)^*\left(\overline{ST^{-1}}\right)-1$ is of trace class.
\item[(A4)] $SJ=JS$ and $TJ=JT$ for some conjugation operator $J$ on $\sH$.
\end{enumerate}

We first recall the Heinz inequality. 
Let $S, T$ be non-negative self-adjoint operators acting in $\sH$.
We write $S \preceq T$ if $\dom(T^{1/2}) \subset \dom(S^{1/2})$ and $\norm{S^{1/2}f} \leq \norm{T^{1/2}f}$ hold for all $f \in \dom(T^{1/2})$.
By a simple limiting argument, $S \leq T$ implies $S \preceq T$.
But the converse is not true. The Heinz inequality asserts that if $S \preceq T$ then $S^p \preceq T^p$ for any $0 < p \leq 1$.
In addition, if $S, T$ are injective, then $S \preceq T$ implies $T^{-1} \preceq S^{-1}$. 
For proofs, see e.g., \cite[Proposition 10.14]{Sc12} and \cite[Corollary 10.12]{Sc12}.

\begin{lem}\label{bounded}
Assume (A1) and (A2).
Then for any $0<p\leq1$, the following statements hold.
\begin{itemize}
\item[(1)] $\dom(S^p)=\dom(T^p)$ and $\dom(S^{-p})=\dom(T^{-p})$,
\item[(2)] $\dom(T^pS^{-p})=\dom(S^{-p})$, $T^pS^{-p}$ is bounded, and $c_1^p\leq\|T^pS^{-p}\|\leq c_2^p$,
\item[(3)] $\dom(T^{-p}S^p)=\dom(S^p)$, $T^{-p}S^p$ is bounded, and $c_2^{-p}\leq\|T^{-p}S^p\|\leq c_1^{-p}$,
\item[(4)] $\dom(S^pT^{-p})=\dom(T^{-p})$, $S^pT^{-p}$ is bounded, and $\overline{S^pT^{-p}}=\left(T^{-p}S^p\right)^*$,
\item[(5)] $\dom(S^{-p}T^p)=\dom(T^p)$, $S^{-p}T^p$ is bounded, and $\overline{S^{-p}T^p}=\left(T^pS^{-p}\right)^*$.
\end{itemize}
\end{lem}

\begin{proof}
By assumption, we have $c_1^2 S^2 \preceq T^2 \preceq c_2^2 S^2$. 
The Heinz inequality implies that $c_1^{2p}S^{2p} \preceq T^{2p} \preceq c_2^{2p}S^{2p}$, 
whence  $\dom(S^p)=\dom(T^p)$ and 
$c_1^p\|S^pf\|\leq\|T^pf\|\leq c_2^p\|S^pf\|$ for all $f\in\dom(S^p)$.
If $h\in\dom(S^{-p})$, then $S^{-p}h\in\dom(S^p)$, and thus we get $c_1^p\|h\|\leq\|T^pS^{-p}h\|\leq c_2^p\|h\|$.
This shows that (2) holds.

On the other hand, $c_1^2S^2\preceq T^2\preceq c_2^2S^2$ implies that $c_2^{-2}S^{-2}\preceq T^{-2}\preceq c_1^{-2}S^{-2}$.
It follows from the Heinz inequality that $c_2^{-2p}S^{-2p}\preceq T^{-2p}\preceq c_1^{-2p}S^{-2p}$, 
whence $\dom(S^{-p})=\dom(T^{-p})$ and $c_2^{-p}\|S^{-p}f\|\leq\|T^{-p}f\|\leq c_1^{-p}\|S^{-p}f\|$ for all $f\in\dom(S^{-p})$.
In particular, (1) holds.
If $h\in\dom(S^p )$, then $S^p h\in\dom(S^{-p})$, and thus we get $c_2^{-p}\|h\|\leq\|T^{-p}S^p h\|\leq c_1^{-p}\|h\|$.
This shows that (3) holds.

To see (4), take any $f\in\dom(T^{-p})$.
Then for any $h\in\dom(S^p )$, we obtain
\[
  \langle S^p h, T^{-p}f\rangle = \langle T^{-p}S^p h, f\rangle = \langle h, \left(T^{-p}S^p \right)^*f\rangle,
\]
which means that $T^{-p}f\in\dom(S^p )$ and that $S^p T^{-p}f = (T^{-p}S^p)^*f$.
Therefore (4) holds.
The same argument shows (5).
This completes the proof.
\end{proof}

\begin{lem}\label{domain}
Assume (A1) and (A2). 
Then the domains of $T^{-1/2}S^{1/2}$, $T^{1/2}S^{-1/2}$, $S^{-1/2}T^{1/2}$, $S^{1/2}T^{-1/2}$ 
contain $\dom(S^{1/2})\cap\dom(S^{-1/2})$,
and they leave $\dom(S^{1/2})\cap\dom(S^{-1/2})$ invariant.
\end{lem}

\begin{proof}
By Lemma \ref{bounded}, $\dom(S^{1/2})\cap\dom(S^{-1/2})=\dom(T^{1/2})\cap\dom(T^{-1/2})$ holds.
Hence the domains of $T^{-1/2}S^{1/2}$, $T^{1/2}S^{-1/2}$, $S^{-1/2}T^{1/2}$, $S^{1/2}T^{-1/2}$ contain $\dom(S^{1/2})\cap\dom(S^{-1/2})$.
Let $f\in\dom(S^{1/2})\cap\dom(S^{-1/2})$.
Then $T^{-1/2}S^{1/2}f \in \dom(T^{1/2}) = \dom(S^{1/2})$. 
Moreover, since $S^{1/2}f \in \dom(S^{-1}) = \dom(T^{-1})$, we get $T^{-1/2}S^{1/2}f \in \dom(T^{-1/2}) = \dom(S^{-1/2})$.
Thus $T^{-1/2}S^{1/2}$ leaves $\dom(S^{1/2})\cap\dom(S^{-1/2})$ invariant.
Similar arguments show that $T^{1/2}S^{-1/2}$, $S^{-1/2}T^{1/2}$, $S^{1/2}T^{-1/2}$ leave $\dom(S^{1/2})\cap\dom(S^{-1/2})$ invariant.
\end{proof}

\begin{lem}\label{X and Y}
Assume (A1) and (A2).
Define operators $X,Y\in\mathcal{B}(\sH)$ by
\begin{align}
 X := \frac{1}{2} \left(\overline{T^{-1/2}S^{1/2}}+\overline{T^{1/2}S^{-1/2}}\right), \qquad 
 Y := \frac{1}{2}\left(\overline{T^{-1/2}S^{1/2}}-\overline{T^{1/2}S^{-1/2}}\right).
  \label{defXY}
\end{align}
Then 
\begin{align*}
 X^* = \frac{1}{2} \left(\overline{S^{1/2}T^{-1/2}}+\overline{S^{-1/2}T^{1/2}}\right), \qquad 
 Y^* = \frac{1}{2} \left(\overline{S^{1/2}T^{-1/2}}-\overline{S^{-1/2}T^{1/2}}\right).  
\end{align*}
In particular, $X, Y, X^*, Y^*$ leave $\dom(S^{1/2})\cap\dom(S^{-1/2})$ invariant.
Moreover they obey the following equalities:
\begin{align*}
 X^*X-Y^*Y &= 1, &  X^*Y-Y^*X &= 0,\\
 XX^*-YY^* &= 1, & -XY^*+YX^* &= 0.
\end{align*}
\end{lem}

\begin{proof}
The first part of the lemma follows from Lemma \ref{bounded} (4) and (5).
The invariance of $\dom(S^{1/2})\cap\dom(S^{-1/2})$ under the actions of $X, Y, X^*, Y^*$ follows from Lemma \ref{domain}.
We next show that $X^*X-Y^*Y=1$.
Let $f\in\dom(S^{1/2})\cap\dom(S^{-1/2})$.
Thanks to Lemma \ref{domain}, we have
\begin{align*}
 (X^*X-Y^*Y)f
 &= \frac{1}{4}\Bigl(S^{1/2}T^{-1}S^{1/2}f+f+f+S^{-1/2}TS^{-1/2}f\Bigr)\\
 &  \quad -\frac{1}{4}\Bigl(S^{1/2}T^{-1}S^{1/2}f-f-f+S^{-1/2}TS^{-1/2}f\Bigr)\\
 &= f.
\end{align*}
By a limiting argument, $X^*X-Y^*Y=1$ holds.
The other three equalities can be proved similarly.
\end{proof}

\begin{lem}\label{Y is HS}
Assume (A1)--(A3). Let $X,Y$ be defined in Lemma \ref{X and Y}.
Then $Y$ is Hilbert-Schmidt.
\end{lem}

\begin{proof}
Let $\{f_m\}_m$ be an orthonormal basis of $\sH$ consisting of vectors in $\dom(S^{1/2})\cap\dom(S^{-1/2})$.
By Lemma \ref{domain} and Lemma \ref{X and Y}, we have
\[
4\langle f_m,Y^*Yf_m\rangle = \langle f_m,S^{1/2}(T^{-1}-S^{-1})S^{1/2}f_m\rangle + \langle f_m,S^{-1/2}(T-S)S^{-1/2}f_m\rangle.
\]
To prove that $Y$ is Hilbert-Schmidt, it is sufficient to show that
\[
\sum_m|\left\langle f_m,S^{1/2}(T^{-1}-S^{-1})S^{1/2}f_m\rangle\right| + \sum_m\left|\langle f_m,S^{-1/2}(T-S)S^{-1/2}f_m\rangle\right|<\infty.
\]
We shall estimate the first term and the second term separately.
Take an orthonormal basis $\{e_n\}_n$ of $\sH$ and a sequence $\{\lambda_n\}_n$ of real numbers so that
\[
 \left(\overline{ST^{-1}}\right)^*\left(\overline{ST^{-1}}\right)-1 = \sum_n\lambda_n|e_n\rangle\langle e_n|.
\]

\textbf{Step 1.} It follows that
\[
 \langle h,(T^2+t^2)^{-1}h-(S^2+t^2)^{-1}h\rangle = \sum_n\lambda_n\langle h,T(T^2+t^2)^{-1}e_n\rangle\langle S(S^2+t^2)^{-1}\overline{S^{-1}T}e_n,h\rangle
\]
for all $h\in\sH$ and $t\in\RR_{>0}$. To see this, note that
\[
 (T^2+t^2)^{-1}h = (T^2+t^2)^{-1}(S^2+t^2)(S^2+t^2)^{-1}h,
\]
and that
\[
(S^2+t^2)^{-1}h = (T^2+t^2)^{-1}(T^2+t^2)(S^2+t^2)^{-1}h,
\]
where the second equality follows from the fact that $\dom(S^2)=\dom(T^2)$.
Hence
\[
 (T^2+t^2)^{-1}h - (S^2+t^2)^{-1}h = (T^2+t^2)^{-1}(S^2-T^2)(S^2+t^2)^{-1}h.
\]
On the other hand, since
\begin{align*}
 T(S^2+t^2)^{-1}h & \in \dom(T)\cap\dom(T^{-1})\\
 & \subset\dom(T^{1/2})\cap\dom(T^{-1/2})
 = \dom(S^{1/2})\cap\dom(S^{-1/2}),
\end{align*}
it follows from Lemma \ref{bounded} (4) and Lemma \ref{domain} that
\[
\left\{\left(\overline{ST^{-1}}\right)^*\left(\overline{ST^{-1}}\right)-1\right\}\cdot T(S^2+t^2)^{-1}h
=(T^{-1}S^2-T)(S^2+t^2)^{-1}h\in\dom(T),
\]
and thus we obtain
\begin{align*}
&T(T^2+t^2)^{-1}\cdot \left\{\left(\overline{ST^{-1}}\right)^*\left(\overline{ST^{-1}}\right)-1\right\}\cdot T(S^2+t^2)^{-1}h\\
&= (T^2+t^2)^{-1}T\cdot \left\{\left(\overline{ST^{-1}}\right)^*\left(\overline{ST^{-1}}\right)-1\right\}\cdot T(S^2+t^2)^{-1}h\\
&= (T^2+t^2)^{-1}(S^2-T^2)(S^2+t^2)^{-1}h.
\end{align*}
Therefore we get
\begin{align*}
&\langle h,(T^2+t^2)^{-1}h - (S^2+t^2)^{-1}h\rangle\\
&= \left\langle h,T(T^2+t^2)^{-1}\cdot \left\{\left(\overline{ST^{-1}}\right)^*
\left(\overline{ST^{-1}}\right)-1\right\}\cdot T(S^2+t^2)^{-1}h\right\rangle\\
&= \sum_n\lambda_n\langle h,T(T^2+t^2)^{-1}e_n\rangle\langle e_n,T(S^2+t^2)^{-1}h\rangle\\
&= \sum_n\lambda_n\langle h,T(T^2+t^2)^{-1}e_n\rangle\langle S(S^2+t^2)^{-1}\overline{S^{-1}T}e_n,h\rangle.
\end{align*}
This completes the proof of Step 1.

\textbf{Step 2.} It follows that $\sum_m|\left\langle f_m,S^{1/2}(T^{-1}-S^{-1})S^{1/2}f_m\rangle\right|<\infty$.\\
Indeed, by Step 1 and the formula $T^{-1}=(2/\pi)\int_{\RR_{>0}}(T^2+t^2)^{-1}\,dt$, we have
\begin{align*}
&\langle f_m,S^{1/2}(T^{-1}-S^{-1})S^{1/2}f_m\rangle = \langle S^{1/2}f_m,(T^{-1}-S^{-1})S^{1/2}f_m\rangle\\
&= \frac{2}{\pi}\int_{\RR_{>0}}\langle S^{1/2}f_m,(T^2+t^2)^{-1}S^{1/2}f_m-(S^2+t^2)^{-1}S^{1/2}f_m\rangle\,dt\\
&= \frac{2}{\pi}\int_{\RR_{>0}}\sum_n
\lambda_n\langle S^{1/2}f_m,T(T^2+t^2)^{-1}e_n\rangle\langle S(S^2+t^2)^{-1}\overline{S^{-1}T}e_n,S^{1/2}f_m\rangle\,dt,
\end{align*}
and hence
\begin{align*}
& \sum_m|\langle f_m,S^{1/2}(T^{-1}-S^{-1})S^{1/2}f_m\rangle|\\
& \leq \frac{2}{\pi}\left(\sum_m\int_{\RR_{>0}}\sum_n|\lambda_n|\cdot\left|\langle S^{1/2}f_m,T(T^2+t^2)^{-1}e_n\rangle\right|^2\,dt\right)^{1/2}\\
& \qquad 
  \times\left(\sum_m\int_{\RR_{>0}}\sum_n|\lambda_n|\cdot\left|\langle S(S^2+t^2)^{-1}\overline{S^{-1}T}e_n,S^{1/2}f_m\rangle\right|^2\,dt\right)^{1/2}\\
& = \frac{2}{\pi}\left(\int_{\RR_{>0}}\sum_n|\lambda_n|\cdot\left\|S^{1/2}T^{-1/2}\cdot T^{3/2}(T^2+t^2)^{-1}e_n\right\|^2\,dt\right)^{1/2}\\
& \qquad 
  \times\left(\int_{\RR_{>0}}\sum_n|\lambda_n|\cdot\left\|S^{3/2}(S^2+t^2)^{-1}\cdot\overline{S^{-1}T}e_n\right\|^2\,dt\right)^{1/2}.
\end{align*}
Let $T=\int_{\RR_{>0}}\lambda\,dE_T(\lambda)$ be the spectral resolution of $T$.
Observe that for any $h\in\sH$,
\begin{align*}
 \int_{\RR_{>0}}\left\|T^{3/2}(T^2+t^2)^{-1}h\right\|^2\,dt
& = \int_{\RR_{>0}}\int_{\RR_{>0}}\frac{\lambda^3}{(\lambda^2+t^2)^2}\,d\|E_T(\lambda)h\|^2\,dt\\
& = \int_{\RR_{>0}}\int_{\RR_{>0}}\frac{\lambda^3}{(\lambda^2+t^2)^2}\,dt\,d\|E_T(\lambda)h\|^2 \\
& = \frac{\pi}{4}\|h\|^2.
\end{align*}
Similarly, 
\[
\int_{\RR_{>0}}\left\|S^{3/2}(S^2+t^2)^{-1}h\right\|^2\,dt =  \frac{\pi}{4}\|h\|^2, \ \ \ \ \ h\in\sH.
\]
Therefore, 
\[
\sum_m\left|\langle f_m,S^{1/2}(T^{-1}-S^{-1})S^{1/2}f_m\rangle\right| 
\leq \frac{1}{2}\|S^{1/2}T^{-1/2}\|\cdot\|S^{-1}T\|\sum_n|\lambda_n|
< \infty.
\]
This completes the proof of Step 2.

\textbf{Step 3.} It follows that $\sum_m|\left\langle f_m,S^{-1/2}(T-S)S^{-1/2}f_m\rangle\right|<\infty$.\\
To see this, note that for each $h\in\sH$, we have
\[
T^2(T^2+t^2)^{-1}h = h-t^2(T^2+t^2)^{-1}h,
\]
and
\[
S^2(S^2+t^2)^{-1}h = h-t^2(S^2+t^2)^{-1}h,
\]
whence
\[
T^2(T^2+t^2)^{-1}h - S^2(S^2+t^2)^{-1}h = -t^2\left\{(T^2+t^2)^{-1}h-(S^2+t^2)^{-1}h\right\}.
\]
By Step 1 and the formula $T=(2/\pi)\int_{\RR_{>0}}T^2(T^2+t^2)^{-1}\,dt$, we get
\begin{align*}
&\langle f_m,S^{-1/2}(T-S)S^{-1/2}f_m\rangle = \langle S^{-1/2}f_m,(T-S)S^{-1/2}f_m\rangle\\
&= \frac{2}{\pi}\int_{\RR_{>0}}\langle S^{-1/2}f_m,T^2(T^2+t^2)^{-1}S^{-1/2}f_m-S^2(S^2+t^2)^{-1}S^{-1/2}f_m\rangle\,dt\\
&= -\frac{2}{\pi}\int_{\RR_{>0}}\langle S^{-1/2}f_m,(T^2+t^2)^{-1}S^{-1/2}f_m-(S^2+t^2)^{-1}S^{-1/2}f_m\rangle t^2\,dt\\
&= -\frac{2}{\pi}\int_{\RR_{>0}}\sum_n
\lambda_n\langle S^{-1/2}f_m,T(T^2+t^2)^{-1}e_n\rangle\langle S(S^2+t^2)^{-1}\overline{S^{-1}T}e_n,S^{-1/2}f_m\rangle t^2\,dt,
\end{align*}
and hence
\begin{align*}
&\sum_m \big| \inner{f_m}{S^{-1/2}(T-S)S^{-1/2}f_m} \big|  \\
&\leq \frac{2}{\pi} \left( \sum_m\int_{\RR_{>0}}\sum_n|\lambda_n|\cdot\left|\langle S^{-1/2}f_m,T(T^2+t^2)^{-1}e_n\rangle\right|^2t^2\,dt\right)^{1/2}\\
& \qquad 
\times\left(\sum_m\int_{\RR_{>0}}\sum_n|\lambda_n|\cdot\left|\langle S(S^2+t^2)^{-1}\overline{S^{-1}T}e_n,S^{-1/2}f_m\rangle\right|^2t^2\,dt\right)^{1/2}\\
&= \frac{2}{\pi}\left(\int_{\RR_{>0}}\sum_n|\lambda_n|\cdot\left\|S^{-1/2}T^{1/2}\cdot T^{1/2}(T^2+t^2)^{-1}e_n\right\|^2t^2\,dt\right)^{1/2}\\
& \qquad 
\times\left(\int_{\RR_{>0}}\sum_n|\lambda_n|\cdot\left\|S^{1/2}(S^2+t^2)^{-1}\cdot\overline{S^{-1}T}e_n\right\|^2t^2\,dt\right)^{1/2}.
\end{align*}
Observe that for any $h\in\sH$,
\begin{align*}
\int_{\RR_{>0}} \left\|T^{1/2}(T^2+t^2)^{-1}h\right\|^2t^2\,dt
& = \int_{\RR_{>0}}\int_{\RR_{>0}}\frac{t^2\lambda}{(\lambda^2+t^2)^2}\,d\|E_T(\lambda)h\|^2\,dt \\
& = \int_{\RR_{>0}}\int_{\RR_{>0}}\frac{t^2\lambda}{(\lambda^2+t^2)^2}\,dt\,d\|E_T(\lambda)h\|^2 \\
& = \frac{\pi}{4}\|h\|^2.
\end{align*}
Similarly, 
\[
\int_{\RR_{>0}}\left\|S^{1/2}(S^2+t^2)^{-1}h\right\|^2t^2\,dt =  \frac{\pi}{4}\|h\|^2, \qquad h\in\sH.
\]
Therefore, 
\[
\sum_m\left| \inner{f_m}{S^{-1/2}(T-S)S^{-1/2}f_m} \right| 
\leq \frac{1}{2} \norm{S^{-1/2}T^{1/2}} \cdot\norm{S^{-1}T} \sum_n |\lambda_n|
< \infty,
\]
whence Step 3 holds.
This finishes the proof of Lemma \ref{Y is HS}.
\end{proof}

Thus, by combining Lemma \ref{X and Y} and Lemma \ref{Y is HS}, we have the following result.
\begin{thm}\label{SP2}
  Assume (A1)--(A4). Then the operators $X$ and $Y$ defined in Lemma \ref{X and Y} satisfy
$(X,Y)\in\mathfrak{Sp}_2$.
\end{thm}


\section{Definition and Self-adjointness of Pair Interaction Hamiltonians}\label{DefModel}
In this section, we consider the Hamiltonian of the pair interaction model defined by
\begin{align}
  H := \dGb(T) + \frac{1}{2} \sum_{n=1}^\infty \lambda_n \PhiS(g_n)^2, \label{hamil}
\end{align}
and prove the self-adjointness.
To make the definition clear, we introduce the following conditions.
\begin{enumerate}
\item[(B1)] $T$ is an injective non-negative self-adjoint operator acting in $\sH$.
\item[(B2)] $\lambda_n\in\RR$ and $g_n\in\dom(T^{1/2})\cap\dom(T^{-1/2})$ for all $n\in\NN$.
\item[(B3)] $\sum_{n=1}^\infty |\lambda_n| \cdot\norm{T^{-1/2}g_n}^2<\infty$.
\item[(B4)] $\sum_{n=1}^\infty |\lambda_n| \cdot\norm{T^{1/2}g_n}^2<\infty$.
\item[(B5)] For some $\vep>0$, the operator inequality
  \begin{equation}\label{B5}
    \one + \sum_{n=1}^\infty \lambda_n \ket{T^{-1/2}g_n} \bra{T^{-1/2}g_n} \geq \vep
  \end{equation}
holds.
\item[(B6)] There exists a conjugation $J$ on $\sH$ such that 
  \begin{align*}
    JTJ = T, \qquad Jg_n = g_n, \qquad n\in\NN.
  \end{align*}
\end{enumerate}

Let us assume (B1)--(B4).
We first remark that the second term of the left-hand side in \eqref{B5} is of trace class.
In deed, the infinite sequence converges absolutely in the trace norm because the trace norm of $\ket{T^{-1/2}g_n} \bra{T^{-1/2}g_n}$ is $\|T^{-1/2}g_n\|^2$.
We next set
\begin{equation*}
 D_j := \sum_{n=1}^\infty |\lambda_n|\cdot\norm{T^{(j-2)/2}g_n}^2, \qquad j=1,2.
\end{equation*}
Then $D_j<\infty$ holds for each $j=1,2$.
Moreover, Lemma \ref{phi2<H0} tells us that
\begin{equation}\label{aiueo2}
 \frac{1}{2}\sum_{n=1}^\infty|\lambda_n|\cdot\bignorm{\PhiS(g_n)^2\Psi}
 \leq D_1 \norm{\dGb(T)\Psi} + D_2\norm{\Psi}
\end{equation}
for all $\Psi\in\dom(\dGb(T))$, and thus the Hamiltonian $H$ defined in \eqref{hamil} is well-defined on $\dom(\dGb(T))$.

For the self-adjointness of $H$, the following is fundamental.

\begin{prop}{\label{saH1}}
  Suppose (B1)--(B4) and $D_1<1$.
 Then the Hamiltonian $H$ is self-adjoint on $\dom(\dGb(T))$, bounded from below,
and essentially self-adjoint on any core of $\dGb(T)$.
\end{prop}

\begin{proof}
  By \eqref{aiueo2}, we have
\begin{align*}
 \Bignorm{\frac{1}{2}\sum_{n=1}^\infty\lambda_n\PhiS(g_n)^2\Psi}
 \leq D_1 \norm{\dGb(T)\Psi} + D_2\norm{\Psi}, \qquad \Psi\in\dom(\dGb(T)).
\end{align*}
Thus the proposition follows from the Kato-Rellich theorem.
\end{proof}

We next show the self-adjointness of $H$ without the condition $D_1<1$.

\begin{lem}{\label{saH2}}
 Let $R$ be an injective non-negative self-adjoint operator acting in $\sH$, 
and let $f_n\in\dom(R^{1/2})\cap \dom(R^{-1/2})$ for all $n\in\NN$.
Then, for any $c>0$ and $N\in\NN$, the inequality
\begin{align}
& \norm{\dGb(R)\Psi}^2 +  \norm{c\sum_{n=1}^N \PhiS(f_n)^2\Psi}^2  \notag \\
& \leq  \bignorm{\big(\dGb(R) + c\sum_{n=1}^N \PhiS(f_n)^2 \big)\Psi}^2 + c\sum_{n=1}^N\norm{R^{1/2}f_n}^2\norm{\Psi}^2 \label{key1}
\end{align}
holds for all $\Psi\in\dom(\dGb(R))$. 
Moreover, $\dGb(R)+\sum_{n=1}^N\PhiS(f_n)^2$ is self-adjoint on $\dom(\dGb(R))$,
and essentially self-adjoint on any core of $\dGb(R)$.
\end{lem}

\begin{proof}
Set $H_I := \sum_{n=1}^N\PhiS(f_n)^2$. Note that \eqref{key1} is equivalent to
\begin{align}
 -2\mathrm{Re} \inner{\dGb(R)\Psi}{H_I\Psi} 
 \leq \sum_{n=1}^N \norm{R^{1/2}f_n}^2 \norm{\Psi}^2. \label{commm}
\end{align}
One can show that (see \cite[Proposition 3.4]{ms05})
\begin{equation}\label{miyaosasaki}
  -2\mathrm{Re}\inner{\dGb(R)\Psi}{\PhiS(f_n)^2\Psi}
 = -2\norm{\dGb(R)^{1/2}\PhiS(f_n)\Psi}^2 + \norm{R^{1/2}f_n}^2\norm{\Psi}^2.
\end{equation}
Thus, by taking a sum over $n$, we get \eqref{commm}, and so \eqref{key1} holds.

We next show the self-adjointness.
Let $H_k:=\dGb(R)+(2/3)^k H_I$ for each $k\in\NN$.
By the same estimate with the proof of Proposition \ref{saH1}, one can find a sufficiently large $k\in\NN$ so that
$H_k$ is self-adjoint on $\dom(\dGb(R))$ and essentially self-adjoint on any core of $\dGb(R)$.
Letting $c=(2/3)^k$ in \eqref{key1}, we obtain
\[
\left\|\frac{1}{2}\cdot\left(\frac{2}{3}\right)^kH_I\Psi\right\|^2
 \leq  \frac{1}{4}\bignorm{H_k\Psi}^2 + \frac{1}{4}\cdot\left(\frac{2}{3}\right)^k\sum_{n=1}^N\norm{R^{1/2}f_n}^2\norm{\Psi}^2
\]
for all $\Psi\in\dom(\dGb(R))$.
Thus the Kato-Rellich theorem implies that 
\[
H_k+\frac{1}{2}\cdot\left(\frac{2}{3}\right)^kH_I = H_{k-1}
\]
is self-adjoint on $\dom(\dGb(R))$ and essentially self-adjoint on any core of $\dGb(R)$.
Repeating this argument $k$ times, we get the desired result.
\end{proof}

\begin{thm}{\label{saH}}
  Suppose (B1)--(B6).
 Then $H$ is self-adjoint on $\dom(\dGb(T))$, bounded from below, 
and essentially self-adjoint on any core of $\dGb(T)$.
\end{thm}

\begin{proof}
Let $0<\vep<1$ be as stated in (B5).
By (B3), we can take a sufficiently large number $N\in\NN$ so that
\begin{equation}\label{vep/4}
\sum_{n=N+1}^\infty|\lambda_n|\cdot\|T^{-1/2}g_n\|^2<\frac{\vep}{4}.
\end{equation}
Fix such an $N$.
It follows from (B5) that
\[
 \vep \leq 1 + \sum_{n=1}^\infty \lambda_n \ket{T^{-1/2}g_n} \bra{T^{-1/2}g_n} 
 \leq 1 + \sum_{n=1}^N \lambda_n \ket{T^{-1/2}g_n} \bra{T^{-1/2}g_n} +\frac{\vep}{4},
\]
and thus we get
\begin{equation}\label{epsilon estimate}
\left(1-\frac{\vep}{2}\right)+\sum_{n=1}^N \lambda_n \ket{T^{-1/2}g_n} \bra{T^{-1/2}g_n} \geq 0.
\end{equation}
Let $\mathscr{M}:=\mathrm{L.h.}\{T^{-1/2}g_1,\cdots,T^{-1/2}g_N\}$, and let $M:=\dim \mathscr{M}$.
Applying Gram-Schmidt to $\{T^{-1/2}g_1,\cdots, T^{-1/2}g_N\}$, we obtain the orthogonal basis $\{e_j\}_{j=1}^M$ of $\mathscr{M}$.
We denote by $P_\mathscr{M}$ the orthogonal projection onto $\mathscr{M}$.
Since $P_\mathscr{M}$ is a linear combination of $\ket{T^{-1/2}g_m}\bra{T^{-1/2}g_n}\ (m,n=1,\cdots,N)$,
the operator $T^{1/2}P_\mathscr{M} T^{1/2}$ is well-defined on $\dom(T^{1/2})$ and bounded.
We set 
\begin{align}
  T_\mathscr{M} := \overline{T^{1/2}P_\mathscr{M} T^{1/2}} 
       = \sum_{j=1}^M \ket{T^{1/2}e_j}\bra{T^{1/2}e_j}, \qquad T_\vep:=T-\left(1-\frac{\vep}{2}\right)T_\mathscr{M}. \label{defTM}
\end{align}
Since $T_\mathscr{M}$ is bounded self-adjoint, $T_\vep$ is self-adjoint on $\dom(T)$.
Furthermore 
\begin{align*}
 0\leq T_\mathscr{M} \leq T, \qquad \frac{\vep}{2}\cdot T \leq T_\vep \leq T
\end{align*}
hold.
In particular, $T_\vep$ is an injective non-negative self-adjoint operator.
It follows from Lemma \ref{bounded} that
\begin{equation}\label{Tvep0}
\dom(T_\vep^{1/2})=\dom(T^{1/2}), \qquad \dom(T_\vep^{-1/2})=\dom(T^{-1/2}),
\end{equation}
and that $T_\vep^{-1/2}T^{1/2}$ is bounded with
\begin{equation}\label{Tvep1}
\|T_\vep^{-1/2}T^{1/2}\|^2\leq\frac{2}{\vep}.
\end{equation}

We first show the operator equality
\begin{equation}\label{T=TE+TM}
\dGb(T)=\dGb(T_\vep)+(1-\vep/2)\dGb(T_\mathscr{M}).
\end{equation} 
Since $\mathscr{M}\subset\dom(T)$, we have $e_j\in\dom(T)$.
Set $\Phi_j := \PhiS(T^{1/2}e_j)$ and $\Pi_j := \PhiS(iT^{1/2}e_j)$ for each $j=1,\cdots,M$.
It holds that
\begin{align}\label{dGbTM}
 \dGb(T_\mathscr{M}) 
& = \sum_{j=1}^M A^*(T^{1/2}e_j) A(T^{1/2}e_j) \notag\\
& = \frac{1}{2}\sum_{j=1}^M  \left( \Phi_j^2 + \Pi_j^2 \right)
    - \frac{1}{2} \sum_{j=1}^M \norm{T^{1/2}e_j}^2 
\end{align}
on $\Ffin(\sH)$.
Note that since $T^{1/2}e_j\in\dom(T^{1/2})\cap\dom(T^{-1/2})$, 
it follows from \eqref{Tvep0} that 
\begin{equation}\label{Tvep1/2}
T^{1/2}e_j\in\dom(T_\vep^{1/2})\cap\dom(T_\vep^{-1/2}).
\end{equation}
This together with Lemma \ref{phi2<H0} implies that $\dom(\dGb(T_\vep))\subset\dom(\Phi_j^2)\cap\dom(\Pi_j^2)$, 
whence for any $\Psi\in\Ffin(\sH)$ and $\Xi\in\dom(\dGb(T_\vep))$, we have
\begin{equation*}
\left\langle\dGb(T_\mathscr{M})\Psi,\Xi \right\rangle 
= \left\langle\Psi,\left[\frac{1}{2}\sum_{j=1}^M  \left( \Phi_j^2 + \Pi_j^2 \right)
    - \frac{1}{2} \sum_{j=1}^M \norm{T^{1/2}e_j}^2\right]\Xi \right\rangle.
\end{equation*}
Since $\Ffin(\sH)$ is a core of $\dGb(T_\mathscr{M})$, 
we get $\dom(\dGb(T_\vep))\subset\dom(\dGb(T_\mathscr{M}))$, and \eqref{dGbTM} holds on $\dom(\dGb(T_\vep))$.
By \eqref{Tvep1/2} and Lemma \ref{saH2}, the operator
\begin{align*}
&\dGb(T_\vep) + \left(1-\frac{\vep}{2}\right)\left[\frac{1}{2}\sum_{j=1}^M  \left( \Phi_j^2 + \Pi_j^2 \right) - \frac{1}{2} \sum_{j=1}^M \norm{T^{1/2}e_j}^2\right]\\
&= \dGb(T_\vep) + \left(1-\frac{\vep}{2}\right)\dGb(T_\mathscr{M})
\end{align*}
is self-adjoint on $\dom(\dGb(T_\vep))$, and essentially self-adjoint on any core of $\dGb(T_\vep)$.
Hence \eqref{T=TE+TM} holds by the following two facts.
First, since $\dom(T)=\dom(T_\vep)$, the subspace $\Ffin(\dom(T))$ is a core of both $\dGb(T)$ and $\dGb(T_\vep)$.
Second, \eqref{T=TE+TM} holds on $\Ffin(\dom(T))$.
This finishes the proof of \eqref{T=TE+TM}.

We next show the self-adjointness of 
\[
H_{\rm fin}:=\dGb(T)+\frac{1}{2}\sum_{n=1}^N \lambda_n \PhiS(g_n)^2 + \frac{1}{2}\left(1-\frac{\vep}{2}\right) \sum_{j=1}^M \norm{T^{1/2}e_j}^2.
\]
Note that, by construction and by (B6), $e_j$ is a real linear combination of the vectors $T^{-1/2}g_1,\cdots,T^{-1/2}g_N$.
Thus, by (B6) again, we have $Je_j=e_j$, and $\langle e_j,T^{-1/2}g_n\rangle\in\RR$.
Then it holds that
\begin{align*}
 \sum_{n=1}^N \lambda_n \PhiS(g_n)^2
 = \sum_{n=1}^N \lambda_n \sum_{j,\ell=1}^M \inner{e_j}{T^{-1/2}g_n} \inner{e_\ell}{T^{-1/2}g_n} \Phi_j \Phi_\ell  
 = \sum_{j,\ell=1}^M G_{j\ell}\Phi_j \Phi_\ell  
\end{align*}
on $\Ffin(\sH)$, where $G:=\sum_{n=1}^N\lambda_n\ket{T^{-1/2}g_n} \bra{T^{-1/2}g_n}$ and $G_{j\ell}:=\inner{e_j}{Ge_\ell}$.
This together with \eqref{dGbTM} leads to the expression
\begin{align}\label{d+N}
&\left(1-\frac{\vep}{2}\right)\dGb(T_\mathscr{M}) + \frac{1}{2}\sum_{n=1}^N \lambda_n \PhiS(g_n)^2 + \frac{1}{2}\left(1-\frac{\vep}{2}\right) \sum_{j=1}^M \norm{T^{1/2}e_j}^2 \notag \\
&= \frac{1}{2}\sum_{j,\ell=1}^M \left[\left(1-\frac{\vep}{2}\right)\delta_{j\ell}+G_{j\ell}\right]\Phi_j \Phi_\ell  + \frac{1}{2}\left(1-\frac{\vep}{2}\right)\sum_{j=1}^M \Pi_j^2
\end{align}
on $\Ffin(\sH)$, where $\delta_{j\ell}$ denotes the Kronecker delta.
By \eqref{epsilon estimate}, the matrix $((1-\vep/2)\delta_{j\ell}+G_{j\ell})_{j,\ell}$ is a non-negative real symmetric matrix, 
so it can be diagonalized by an orthogonal matrix.
We write the right-hand side of \eqref{d+N} as
\begin{equation*}
  \frac{1}{2}\sum_{j,\ell=1}^M \left[\left(1-\frac{\vep}{2}\right)\delta_{j\ell}+G_{j\ell}\right]\Phi_j \Phi_\ell 
+ \frac{1}{2}\left(1-\frac{\vep}{2}\right)\sum_{j=1}^M \Pi_j^2
= \sum_{j=1}^{2M} \PhiS(f_j)^2
\end{equation*}
on $\Ffin(\sH)$,
where $f_j$ is a linear combination of $T^{1/2}e_1,\cdots,T^{1/2}e_M$.
By \eqref{Tvep0} and \eqref{Tvep1/2}, all $g_n$ and $f_j$ are in $\dom(T_\vep^{1/2})\cap\dom(T_\vep^{-1/2})$, and hence
\begin{equation*}
\left(1-\frac{\vep}{2}\right)\dGb(T_\mathscr{M}) + \frac{1}{2}\sum_{n=1}^N \lambda_n \PhiS(g_n)^2 + \frac{1}{2}\left(1-\frac{\vep}{2}\right) \sum_{j=1}^M \norm{T^{1/2}e_j}^2
= \sum_{j=1}^{2M} \PhiS(f_j)^2
\end{equation*}
holds on $\dom(\dGb(T_\vep))$.
Now Lemma \ref{saH2} tells us that 
\begin{align*}
&\dGb(T_\vep)+\sum_{j=1}^{2M} \PhiS(f_j)^2\\
&=\dGb(T_\vep) + \left(1-\frac{\vep}{2}\right)\dGb(T_\mathscr{M}) + \frac{1}{2}\sum_{n=1}^N \lambda_n \PhiS(g_n)^2 + \frac{1}{2}\left(1-\frac{\vep}{2}\right) \sum_{j=1}^M \norm{T^{1/2}e_j}^2\\
&=H_{\rm fin}
\end{align*}
is self-adjoint on $\dom(\dGb(T_\vep))=\dom(\dGb(T))$ and bounded from below.
Moreover
\begin{equation}\label{12345}
\|\dGb(T_\vep)\Psi\|^2  \leq \|H_{\rm fin}\Psi\|^2 + \sum_{j=1}^{2M}\|T_\vep^{1/2}f_j\|^2\|\Psi\|^2
\end{equation}
holds for all $\Psi\in\dom(\dGb(T_\vep))$.

Finally, we prove the self-adjointness of $H$. 
From \eqref{vep/4}, \eqref{Tvep1}, \eqref{12345} and Lemma \ref{phi2<H0}, we have
\begin{align*}
&\frac{1}{2}\sum_{n=N+1}^\infty|\lambda_n|\cdot\|\PhiS(g_n)^2\Psi\|\\
&\leq \sum_{n=N+1}^\infty|\lambda_n|\cdot\|T_\vep^{-1/2}g_n\|^2\|\dGb(T_\vep)\Psi\| + \sum_{n=N+1}^\infty\|g_n\|^2\|\Psi\|\\
&\leq \frac{1}{2}\|H_{\rm fin}\Psi\| 
+\left(\frac{1}{2}\sqrt{\sum_{j=1}^{2M}\|T_\vep^{1/2}f_j\|^2}+\sum_{n=N+1}^\infty\|g_n\|^2\right)\|\Psi\|
\end{align*}
for all $\Psi\in\dom(\dGb(T_\vep))$.
Hence the Kato-Rellich theorem implies that $H$ is self-adjoint on $\dom(\dGb(T))$ and bounded from below.
By \eqref{aiueo2}, $H$ is essentially self-adjoint on any core of $\dGb(T)$.
This completes the proof.
\end{proof}


\section{Diagonalization of Pair Interaction Models}\label{Diagonalization}
In this section, we suppose (B1)--(B5) in the previous section.
Let 
\begin{align}
W := \sum_{n=1}^\infty \lambda_n \ket{T^{1/2}g_n}\! \bra{T^{1/2}g_n}.   \label{Wdef}
\end{align}
Then $W$ is of trace class. 
We consider the operator
\begin{align*}
  h_\mathrm{p} := T^2 + W.
\end{align*}
\begin{lem}\label{bdhp}
  Suppose (B1)--(B5). 
Let $\vep>0$ be given in (B5).
Then
\begin{align*}
    c_1^2 h_\mathrm{p} \leq T^2 \leq c_2^2h_\mathrm{p}
\end{align*}
holds with $c_1 := (1+D_1)^{-1/2}$ and $c_2 := \vep^{-1/2}$.
\end{lem}

\begin{proof}
For $v\in\dom(T^2)$, we have
\begin{align*}
 \inner{v}{h_\mathrm{p}v}
& = \inner{v}{T^2v}+\inner{Tv}{T^{-1}WT^{-1}Tv} \\
 &\leq \inner{v}{T^2v}+\sum_{n=1}^\infty|\lambda_n|  \norm{T^{-1/2}g_n}^2 \norm{Tv}^2 \\
 & = (1+D_1)\inner{v}{T^2v},
\end{align*}
which implies $c_1^2h_\mathrm{p}\leq T^2$. 
By (B5), we have $T^{-1}WT^{-1} \geq \vep-1$. Thus we have
\begin{align*}
 \inner{v}{h_\mathrm{p}v}
 \geq \inner{v}{T^2v} + \inner{Tv}{(\vep-1)Tv} = \inner{v}{\vep T^2v},
\end{align*}
which implies $c_2^2 h_\mathrm{p}\geq T^2$.
\end{proof}
Set
\begin{align}
S:=h_\mathrm{p}^{1/2} = \Big(T^2+\sum_{n=1}^\infty \lambda_n \ket{T^{1/2}g_n}\!\bra{T^{1/2}g_n}\Big)^{1/2}.  \label{defS}
\end{align}
\begin{lem}\label{Bogo OK}
  Suppose (B1)--(B6). Then $S$ and $T$ satisfy conditions (A1)--(A4) in Section \ref{construct}.
In particular, the bounded operators $X$ and $Y$ defined in Lemma \ref{X and Y} satisfy 
$(X,Y)\in\mathfrak{Sp}_2$.
\end{lem}

\begin{proof}
The assumption on $T$ and Lemma \ref{bdhp} lead to (A1) and (A2).
(B6) implies (A4). Thus it is enough to check (A3).
Set $L:=\left(\overline{ST^{-1}}\right)^*\left(\overline{ST^{-1}}\right)-1$.
It follows that
\begin{align*}
L\supset T^{-1}S^2T^{-1}-1=T^{-1}(T^2+W)T^{-1}-1= T^{-1}WT^{-1}\big|_{\dom(T)\cap\dom(T^{-1})}.
\end{align*}
The closure of the right-hand side is of trace class, and thus so is $L$.
Hence (A3) holds.
Therefore, by using Theorem \ref{SP2}, $(X,Y)\in\mathfrak{Sp}_2$ follows.
\end{proof}

The following is one of the main theorems in this paper.
\begin{thm}\label{main}
  Suppose (B1)--(B6). Let $H$ be the Hamiltonian defined in \eqref{hamil}. 
Let $(X,Y)\in\mathfrak{Sp}_2$ be the operators defined in Lemma \ref{Bogo OK}, and 
let $U$ be the corresponding Bogoliubov transformation so that \eqref{BT} holds.
Then
\begin{align}
  UHU^* = \dGb(S)+E  \label{ExpDiag}
\end{align}
for some $E\in\RR$. More explicitly,
\begin{align*}
  UHU^* = \dGb\Big(\Big( T^2+\sum_{n=1}^\infty\lambda_n\ket{T^{1/2}g_n}\!\bra{T^{1/2}g_n}\Big)^{1/2}\Big) + E.
\end{align*}
\end{thm}

\begin{rem}
Let $\sH:=L^2(\RR^d,dx)$ and $T:=\sqrt{-\Delta+m^2}$, where $m\geq 0$ is a constant.
Under the conditions (B1)--(B6), Theorem \ref{main} gives a unitary equivalence 
\[
UHU^*= \dGb \left(\sqrt{-\Delta+m^2+\sum_{n=1}^{\infty}\lambda_n|T^{1/2}g_n\rangle\langle T^{1/2}g_n|}\right) + E.
\]
Thus, in this case, the Hamiltonian $H$ is essentially described by a trace class perturbation of the free Schr\"odinger operator.
\end{rem}

The proof of Theorem \ref{main} will be completed at the end of this section.
The following is the key lemma for computing the Bogoliubov transformation.
\begin{lem}{\label{interP}}
Suppose (B1)--(B6). Then, $X\dom(T)\subset \dom(T)$, 
$Y\dom(T)\subset \dom(T)$ and the equations 
\begin{align*}
  & TX = XS - \frac{1}{2}W_0(X-Y), \\
  & TY = -YS + \frac{1}{2}W_0(X-Y)
\end{align*}
hold on $\dom(T)$, where $W_0$ is a bounded operator defined by
\begin{align}
  W_0 := \overline{T^{-1/2}WT^{-1/2}} = \sum_{n=1}^\infty \lambda_n \ket{g_n} \! \bra{g_n}, \label{defW0}
\end{align}
\end{lem}

\begin{proof}
From (B3) and (B4), the boundedness of $W_0$ follows.
Note that $\dom(T^2)=\dom(S^2)$ and $\dom(T^p)=\dom(S^p)$ for all $|p|\leq 1$ (Lemma \ref{bounded} (1)).
For $v\in\dom(S^{1/2})\cap \dom(S^{-1/2})$, we have $S^{-1/2}v \in \dom(S)=\dom(T)$.
We also note that $Xv\in\dom(T^{1/2})\cap\dom(T^{-1/2})$ by Lemma \ref{X and Y}. 
For all $u\in \dom(T^2)\cap\dom(T^{-1/2})$ and $v \in \dom(T)\cap\dom(T^{-1/2})$, we have
\begin{align*}
  2\inner{Tu}{Xv} 
& = \inner{Tu}{(T^{-1/2}S^{1/2}+T^{1/2}S^{-1/2})v} \\
& = \inner{T^{1/2}u}{S^{1/2}v} + \inner{T^2T^{-1/2}u}{S^{-1/2}v} \\
& = \inner{T^{1/2}u}{S^{1/2}v} + \inner{(S^2-W)T^{-1/2}u}{S^{-1/2}v} \\
& = \inner{u}{T^{1/2}S^{-1/2}Sv} + \inner{S^{1/2}T^{-1/2}u}{Sv} - \inner{WT^{-1/2}u}{S^{-1/2}v} \\
& = \inner{u}{\left(T^{1/2}S^{-1/2}+\overline{T^{-1/2}S^{1/2}}\right)Sv} - \inner{WT^{-1/2}u}{S^{-1/2}v} \\
& = \inner{u}{2XSv} - \inner{u}{W_0T^{1/2}S^{-1/2}v} \\
& = \inner{u}{2XSv} - \inner{u}{W_0(X-Y)v}.
\end{align*}
Since $\dom(T^2)\cap\dom(T^{-1/2})$ is a core for $T$, we have $Xv\in\dom(T)$ and 
\begin{align}
  TXv = XSv - \frac{1}{2}W_0(X-Y)v.  \label{mada}
\end{align}
For any $v\in\dom(T)$, there exist $v_n \in \dom(T)\cap \dom(T^{-1/2})$ such that 
$v_n\to v$ and $Tv_n\to Tv$ as $n\to\infty$. Then $Sv_n\to Sv$ and hence
$Xv_n$ is Cauchy by \eqref{mada}. Thus $Xv_n\to Xv\in\dom(T)$ and
\eqref{mada} holds for all $v\in\dom(T)$.
Similarly, we have that $Yv\in\dom(T)$ and $TYv=-YSv+(1/2)W_0(X-Y)$ hold.
\end{proof}
\begin{proof}[Proof of Theorem \ref{main}]
In order to prove Theorem \ref{main}, we apply Theorem \ref{diag prop} to $H$.
It is enough to check the conditions (i)--(v).
We set 
\begin{align*}
 \sD_1 := \dom(T).  
\end{align*}
Then $\sD_1$ is dense, and since $\dom(H)=\dom(\dGb(T))$ by Theorem \ref{saH}, $\dom(H)$ contains $\Ffin(\dom(T))$.
Thus (i) holds.

Note that Lemma \ref{bdhp} leads to $c_1^2 S^2 \preceq T^2 \preceq c_2^2S^2$, 
and hence $c_1^2 \left(S^{(n)}\right)^2 \preceq \left(T^{(n)}\right)^2 \preceq c_2^2\left(S^{(n)}\right)^2$, where $S^{(n)}$ is defined in \eqref{S(n)}.
This fact implies that $\dom(\dGb(T))=\dom(\dGb(S))$, and thus $\dom(H)\subset\dom(\dGb(S)^{1/2})$.
Therefore (ii) holds.

Next we show that (iii) holds with
\begin{align*}
 \sD:=\dom(S^2)\cap\dom(S^{-1/2}).  
\end{align*}
Clearly, $\sD\subset \dom(S)$ and $e^{itS}\sD = \sD$ for all $t\in\RR$.
For all $f\in\sD$ we have $f,Sf \in \dom(S^{1/2})\cap\dom(S^{-1/2})$.
Since $F(f)=Xf+JYf$ and $X,Y$ leave $\dom(S^{1/2})\cap\dom(S^{-1/2})$ invariant (Lemma \ref{X and Y}),
we have that $F(f),F(Sf)\in\dom(S^{-1/2})$. Thus (iii) follows.

Next we show (iv). For $f\in\sD$, we have
\begin{align*}
  \norm{S^{-1/2}X(\vep^{-1}(e^{i\vep S}-1)-iS)f} 
\leq \norm{S^{-1/2}XS^{1/2}}\cdot \norm{(\vep^{-1}(e^{i\vep S}-1)-iS)S^{-1/2}f}.
\end{align*}
Since $S^{-1/2}f\in\dom(S)$, the right-hand side converges to zero as $\vep\to 0$. 
Similarly, we have
\begin{align*}
  \norm{S^{-1/2} JY(\vep^{-1}(e^{i\vep S}-1)-iS)f} 
&\leq \norm{S^{-1/2}YS^{1/2}}\cdot \norm{(\vep^{-1}(e^{i\vep S}-1)-iS)S^{-1/2}f} \\
& \to 0 \qquad (\vep\to 0).
\end{align*}
 Hence (iv) holds.

Finally we show the last condition (v). 
Let $f\in\sD $.
By Lemma \ref{interP}, we obtain 
\[
B(f)\Psi,\ B^*(f)\Psi\in\Ffin(\dom(T))\subset\dom(H)
\]
for all $\Psi\in\Ffin(\dom(T^2))$.
Keeping this in mind, we first show that
\begin{align*}
  \langle\Phi,[H,B(f)]\Psi\rangle = \langle\Phi,-B(Sf)\Psi\rangle, \qquad\Phi,\Psi\in\Ffin(\dom(T^2)).
\end{align*} 
We have
\begin{align*}
 [\dGb(T),B(f)] 
= [\dGb(T),A(Xf)+A^*(JYf)] 
= A(-TXf)+A^*(TJYf) 
\end{align*}
on $\Ffin(\dom(T^2))$.
By using Lemma \ref{interP}, we get
\begin{align*}
&  A(-TXf) = -A(XSf) + \frac{1}{2}A(W_0(X-Y)f), \\
&  A^*(TJYf) = A^*(TYJf) = -A^*(YSJf) + \frac{1}{2}A^*(W_0(X-Y)Jf).
\end{align*}
Hence
\begin{align*}
 [\dGb(T),B(f)] 
 = -B(Sf) + \frac{1}{2}A(W_0(X-Y)f) + \frac{1}{2}A^*(W_0(X-Y)Jf)
\end{align*}
holds on $\Ffin(\dom(T^2))$. On the other hand, it holds that
\begin{align*}
&\frac{1}{2}\sum_{n=1}^\infty\lambda_n [\PhiS(g_n)^2,B(f)] \\
&= \frac{1}{2}\sum_{n=1}^\infty\lambda_n [\PhiS(g_n)^2,A(Xf)+A^*(JYf)] \\
&= \frac{\sqrt{2}}{2} \sum_{n=1}^\infty\lambda_n \Big( -\inner{Xf}{g_n}\PhiS(g_n) + \inner{g_n}{JYf}\PhiS(g_n) \Big) \\
&= \frac{\sqrt{2}}{2} \sum_{n=1}^\infty\lambda_n \Big( -\inner{Xf}{g_n}\PhiS(g_n) + \inner{Yf}{Jg_n}\PhiS(g_n) \Big) \\
&= -\frac{1}{2} \sum_{n=1}^\infty\lambda_n \inner{(X-Y)f}{g_n} (A(g_n)+A^*(g_n)) \\
&= -\frac{1}{2}\sum_{n=1}^\infty\lambda_n \Big( A(\inner{g_n}{(X-Y)f}g_n) 
   + A^*(\inner{g_n}{(X-Y)Jf}g_n) \Big) \\
&= -\frac{1}{2}A(W_0(X-Y)f)-\frac{1}{2}A^*(W_0(X-Y)Jf)
\end{align*}
on $\Ffin(\dom(T^2))$. Thus, we have
\begin{equation*}
  \inner{\Phi}{[H,B(f)]\Psi} = -\inner{\Phi}{B(Sf)\Psi}, 
  \qquad {\Phi,\Psi} \in\Ffin(\dom(T^2)).
\end{equation*}
By taking the complex conjugation, we get
\begin{equation*}
  \inner{\Phi}{[H,B^*(f)]\Psi} = \inner{\Phi}{B^*(Sf)\Psi}, \qquad {\Phi,\Psi} \in\Ffin(\dom(T^2)).
\end{equation*}
Simple limiting arguments, together with \eqref{dominc1}, \eqref{dominc2} and the closed graph theorem, implies (v).
Now we have checked all conditions (i)--(v) in Theorem \ref{diag prop}.
As a consequence of Theorem \ref{diag prop}, the unitarily equivalence \eqref{ExpDiag} is established.
\end{proof}

\section{Ground State Energy}{\label{Sec:gse}}
In this section, we give an explicit expression for the ground state energy of $H$.

\begin{thm}\label{gse}
Assume (B1)--(B6). Let $H$ and $S$ be Hamiltonians defined by \eqref{hamil} and \eqref{defS}, respectively.
Then, $\overline{S-T}$ is of trace class, and the ground state energy $E$ of $H$ has the form
\begin{align}
  E = \frac{1}{2} \mathrm{tr}(\overline{S-T}). \label{eq:gse}
\end{align}
\end{thm}

We prepare the next lemma.
\begin{lem}{\label{tr lem}}
Assume (B1)--(B6).
For $p,q$ with $-1/2\leq p,q\leq 1/2$, the operator $T^p(S-T)T^q$ is bounded and $\overline{T^p(S-T)T^q}$ is of trace class. 
\end{lem}
\begin{proof}
We show the lemma only for the case $-1/2\leq q\leq 0 \leq p\leq 1/2$.
The proofs of the other cases are similar.
By Lemma \ref{bounded}, for $u\in \dom(T^2)\cap \dom(T^{-1/2})$, the operation $T^p(S-T)T^qu$ is well-defined.
  By the same estimate as in the proof of Lemma \ref{Y is HS}, we have, for $u, v\in \dom(T^2)\cap \dom(T^{-1/2})$, that
\begin{align*}
& |\inner{u}{T^p(S-T)T^q v}|  \\
&\leq \frac{2}{\pi}\int_0^\infty \big| \inner{u}{T^p ((T^2+t^2)^{-1}-(S^2+t^2)^{-1})T^q v} \big| t^2dt \\
&= \frac{2}{\pi}\int_0^\infty \big| \inner{u}{T^p (T^2+t^2)^{-1}W(S^2+t^2)^{-1} T^q v} \big| t^2dt \\
&= \frac{2}{\pi}\int_0^\infty \big| \inner{T^{1/2}(T^2+t^2)^{-1}u}{T^{p-1/2} WS^{q-1/2} S^{1/2}(S^2+t^2)^{-1} S^{-q}T^q v} \big| t^2dt \\
&\leq \frac{2}{\pi} \norm{T^{p-1/2} WS^{q-1/2}} \int_0^\infty \norm{T^{1/2}(T^2+t^2)^{-1}u} \norm{S^{1/2}(S^2+t^2)^{-1} S^{-q}T^q v}  t^2dt \\
&\leq \frac{1}{2} \norm{T^{p-1/2} WS^{q-1/2}} \norm{u} \norm{S^{-q}T^q} \norm{v},
\end{align*}
where we have used the fact that $\int_0^\infty \norm{T^{1/2}(T^2+t^2)^{-1}u}^2t^2dt = (\pi/4)\norm{u}^2$.
We note that $S^{-q}T^q$ is bounded by Lemma \ref{bounded}.
By conditions (B3) and (B4), we have
\begin{align*}
 \norm{T^{p-1/2} WS^{q-1/2}} 
& \leq \norm{T^{p-1/2} WT^{q-1/2}} \norm{T^{1/2-q}S^{q-1/2}} \notag \\
& \leq \norm{T^{1/2-q}S^{q-1/2}} \sum_{n=1}^\infty |\lambda_n| \norm{T^pg_n} \norm{T^qg_n} < \infty, 
\end{align*}
where we have used the condition $0\leq 1/2 - q \leq 1$ and Lemma \ref{bounded}.
Thus $T^p(S-T)T^q$ is a bounded operator. Next, we show that its closure is of trace class.
By a limiting argument, one has
\begin{align}
& \inner{u}{\overline{T^p(S-T)T^q}v} \notag\\
& = \frac{2}{\pi} \int_0^\infty \inner{u}{(T^2+t^2)^{-1} \overline{T^{p} W S^{q}} (S^2+t^2)^{-1} \overline{S^{-q}T^q}v} t^2dt,\label{TSTT}
\end{align}
for all $u,v\in\sH$.
Let $\{e_n\}_n, \{f_n\}_n$ be orthonormal bases. In order to prove 
the trace property, it is enough to show that 
\begin{align}
  \sum_{n} \Big| \inner{e_n}{\overline{T^p(S-T)T^q}f_n} \Big| \leq C \label{bds-t}
\end{align}
with some constant $C$ independent of $\{e_n\}_n$ and $\{f_n\}_n$ (see \cite[Proposition 3.6.5]{Si15}).
By \eqref{TSTT} and the definition of $W$, we have
\begin{align*}
& \sum_{n} \Big| \inner{e_n}{\overline{T^p(S-T)T^q}f_n} \Big| \\
&\leq \frac{2}{\pi} \sum_{m} | \lambda_m | \int_0^\infty  \sum_{n} \Big| \inner{e_n}{T^p(T^2+t^2)^{-1}T^{1/2}g_m} \! \inner{S^qT^{1/2}g_m}{(S^2+t^2)^{-1}\overline{S^{-q}T^q}f_n} \Big| t^2 dt \\
&\leq \frac{2}{\pi} \sum_{m} | \lambda_m | \int_0^\infty \norm{T^{p+1/2}(T^2+t^2)^{-1} g_m} \norm{T^qS^{-q}} \norm{(S^2+t^2)^{-1}S^qT^{1/2}g_m} t^2 dt \\
&\leq \frac{2}{\pi} \norm{T^qS^{-q}} \sum_{m} | \lambda_m | \bigg( \int_0^\infty  \norm{T^{p+1/2}(T^2+t^2)^{-1} g_m}^2 t^2dt \bigg)^{1/2} \\
&\qquad \times      \bigg( \int_0^\infty \norm{(S^2+t^2)^{-1}S^q T^{1/2}g_m}^2 t^2 dt  \bigg)^{1/2} \\
&= \frac{2}{\pi} \norm{T^qS^{-q}} \sum_{m} | \lambda_m | \bigg( \frac{\pi}{4}\norm{T^pg_m}^2  \bigg)^{1/2} \bigg( \frac{\pi}{4} \norm{S^{q-1/2}T^{1/2}g_m}^2 \bigg)^{1/2} \\
&= \frac{1}{2} \norm{T^qS^{-q}} \norm{S^{q-1/2}T^{1/2-q}} \sum_{m} | \lambda_m | \norm{T^pg_m} \norm{T^qg_m}  < \infty.
\end{align*}
Thus, \eqref{bds-t} holds, which implies that $\overline{T^p(S-T)T^q}$ is of trace class.
\end{proof}
\begin{proof}[Proof of Theorem \ref{gse}]
It follows from Proposition \ref{gse1} that $YS^{1/2}$ is Hilbert-Schmidt and that
\begin{equation}
  E = \inner{\Omega}{H\Omega} - \|\overline{YS^{1/2}}\|_{\rm HS}^2
=\frac{1}{4} \sum_{n=1}^\infty \lambda_n \norm{g_n}^2 - \mathrm{tr} (\overline{YSY^*}). \label{gse2}
\end{equation}
Note that since $Y^* \dom(S)\subset \dom(S)$ by Lemma \ref{interP}, $YSY^*$ is densely defined, and thus its closure is of trace class.

Next, we compute $\mathrm{tr}(\overline{YSY^*})$. 
Since $\dom(T^k)=\dom(S^k)\ (k=1,2)$, we have
\begin{align*}
  \overline{YSY^*}
 &= \frac{1}{4} (T^{-1/2}S^{1/2}-T^{1/2}S^{-1/2})(S^{3/2}T^{-1/2}-S^{1/2}T^{1/2}) \\
 &= \frac{1}{4} \big( T^{-1/2}WT^{-1/2} +2T-T^{-1/2}ST^{1/2}-T^{1/2}ST^{-1/2} \big)  \\
 &= \frac{1}{4} \big( T^{-1/2}WT^{-1/2} + T^{-1/2}(T-S)T^{1/2} +  T^{1/2}(T-S)T^{-1/2} \big)
\end{align*}
on $\dom(T^2)\cap \dom(T^{-1/2})$.
Thus, by Lemma \ref{tr lem}, we obtain
\begin{align*}
  \mathrm{tr}(\overline{YSY^*})
& = \frac{1}{4} \mathrm{tr}( \overline{T^{-1/2}WT^{-1/2}} )
    +\frac{1}{4} \mathrm{tr} ( \overline{T^{-1/2}(T-S)T^{1/2}} +  \overline{T^{1/2}(T-S)T^{-1/2}} ).\\
& = \frac{1}{4} \mathrm{tr}( \overline{T^{-1/2}WT^{-1/2}} ) - \frac{1}{2} \mathrm{Re} ~ \mathrm{tr} ( \overline{T^{-1/2}(S-T)T^{1/2}} ).
\end{align*}
By the definition of $W$, we have
\begin{align*}
  \frac{1}{4} \mathrm{tr}( \overline{T^{-1/2}WT^{-1/2}} )
 = \frac{1}{4} \sum_{n=1}^\infty \lambda_n\norm{g_n}^2.
\end{align*}
Combining this fact with \eqref{gse2}, we get
\begin{align*}
  E &=  \frac{1}{2} \mathrm{Re} ~ \mathrm{tr} ( \overline{T^{-1/2}(S-T)T^{1/2}} ).
\end{align*}
By using Lemma \ref{tr lem}, the range of $\overline{(S-T)T^{1/2}}$ is contained in $\dom(T^{-1/2})$, and hence
\begin{align}
 \overline{T^{-1/2}(S-T)T^{1/2}} = T^{-1/2}\overline{(S-T)T^{1/2}}.   \label{tstt}
\end{align}
Since $\overline{(S-T)T^{1/2}}$ is of trace class (Lemma \ref{tr lem}), it has a canonical decomposition
\begin{align*}
  \overline{(S-T)T^{1/2}} = \sum_m \mu_m \ket{e_m} \bra{f_m},
\end{align*}
where $\mu_m > 0$ and $\{e_m\}_m$, $\{f_m\}_m$ are orthonormal systems.
Note that $\overline{(S-T)T^{1/2}}f_m = \mu_m e_m$ and $\big(\overline{(S-T)T^{1/2}}\big)^* e_m = \mu_m f_m$.
In particular, it follows that $(\{e_m\}_m)^\perp\subset\ker(T^{1/2}(\overline{S-T}))=\ker(\overline{S-T})$ 
and that $(\{f_m\}_m)^\perp\subset\ker(\overline{(S-T)T^{1/2}})$.
By \eqref{tstt}, we have $e_m \in \dom(T^{-1/2})$.
Summing up the above arguments, we get
\begin{align*}
E &= \frac{1}{2} \mathrm{Re} \sum_m\inner{f_m}{T^{-1/2}\overline{(S-T)T^{1/2}}f_m} \\
&= \frac{1}{2} \mathrm{Re} \sum_m \inner{f_m}{T^{-1/2} \mu_m e_m} \\
&= \frac{1}{2} \mathrm{Re} \sum_m \inner{\mu_m f_m}{T^{-1/2} e_m} \\
&= \frac{1}{2} \mathrm{Re} \sum_m \inner{\big(\overline{(S-T)T^{1/2}}\big)^* e_m}{T^{-1/2} e_m} \\
&= \frac{1}{2} \mathrm{Re} \sum_m \inner{T^{1/2}(\overline{S-T}) e_m}{T^{-1/2} e_m} \\
&= \frac{1}{2} \mathrm{Re} \sum_m \inner{(\overline{S-T})e_m}{e_m} \\
&= \frac{1}{2} \mathrm{tr}(\overline{S-T}).
\end{align*}
Therefore \eqref{eq:gse} holds.
\end{proof}


\section{Examples}{\label{examples}}
In this section, we apply our results to several concrete Hamiltonians.
Before going to examples, we recall some notations.
Let $L^2(\RR^d)=L^2(\RR^d,dx)$, where $dx$ is the $d$-dimensional Lebesgue measure.
Put $p_j:=-i\partial/\partial x_j$ for each $j=1,\cdots,d$ that acts in $L^2(\RR^d)$.
In the case of $d=1$, we write $p_1$ as $p$ for notational simplicity.
We identify each Borel function on $\RR^d$ with the corresponding multiplication operator on $L^2(\RR^d)$.

For two complex Hilbert spaces $\sH,\mathscr{K}$, we use the natural isomorphism
$\Fb(\sH)\otimes\Fb(\mathscr{K})=\Fb(\sH\oplus\mathscr{K})$.
For the details, see e.g., \cite[Section 5.20]{A18}.


\subsection{The Single Pair Interaction Model}\label{pair interaction model}
Let $\sH$ be a separable complex Hilbert space, and let $T$ be an injective non-negative self-adjoint operator acting in $\sH$.
The Hamiltonian $H$ of the single pair interaction model is defined as follows:
\[
H:={\rm d}\Gamma_{\rm b}(T) + \frac{\lambda}{2}\Phi_{\rm S}(g)^2,
\] 
where $\lambda\in\RR$ is a constant and $g$ is a vector in ${\rm dom}(T^{1/2})\cap{\rm dom}(T^{-1/2})$.
Note that $H$ acts in $\Fb(\sH)$.
This Hamiltonian was mathematically studied by Asahara and Funakawa \cite{AF}.
We first note that $H$ is of the form (\ref{hamil}).
We next check that $H$ satisfies the conditions (B1)--(B6).
However, (B1)--(B4) are trivial, and (B6) is automatic by the following lemma.

\begin{lem}\label{conjugation lemma}
Let $A$ be a self-adjoint operator acting in a separable complex Hilbert space $\mathscr{K}$, and let $h\in\mathscr{K}$ be arbitrary.
Then there exists a conjugation $J$ on $\mathscr{K}$ such that $JAJ=A$ and that $Jh=h$.
\end{lem}

\begin{proof}
There are a measure space $(M,\mu)$ and a unitary operator $U:\mathscr{K}\to L^2(M,\mu)$ such that $UAU^*$ is the multiplication operator by a real valued function on $M$. 
Define a function $v:M\to\CC$ by
\[
  v(m) := \begin{cases}
    (Uh)(m)/|(Uh)(m)|, & {\rm if}\ \  (Uh)(m)\not=0, \\
    1, &  {\rm if}\ \ (Uh)(m)=0,
  \end{cases}\ \ \ \ \ m\in M.
\]
We denote by $V$, the multiplication operator by $v$.
Note that $V$ is a unitary operator on $L^2(M,\mu)$.
Let $J_0$ be a conjugation on $L^2(M,\mu)$ defined by
\[
(J_0F)(m) := \overline{F(m)},\ \ \ \ \ F\in L^2(M,\mu),\ m\in M.
\]
Then $J:=U^*VJ_0V^*U$ satisfies $JAJ=A$ and $Jh=h$.
\end{proof}

We now state the main result of this subsection.

\begin{theorem}\label{diagonalization of SPIM}
The Hamiltonian $H$ satisfies the condition (B5) if and only if 
\[
1+\lambda\|T^{-1/2}g\|^2>0.
\]
In this case, $H$ is self-adjoint, and essentially self-adjoint on any core of ${\rm d}\Gamma_{\rm b}(T)$.
Furthermore, there exists a unitary operator $U$ on $\Fb(\sH)$ such that 
\[
UHU^*={\rm d}\Gamma_{\rm b}\left(\sqrt{T^2+\lambda|T^{1/2}g\rangle\langle T^{1/2}g|}\right) + E
\]
with
\[
E=\frac{1}{2}\mathrm{tr}\left(\,\overline{\sqrt{T^2+\lambda|T^{1/2}g\rangle\langle T^{1/2}g|}-T}\,\right).
\]
\end{theorem}

\begin{proof}
We may assume that $g\not=0$.
If $H$ satisfies (B5), then 
\[
\left\langle T^{-1/2}g, \Bigl(1+\lambda|T^{-1/2}g\rangle\langle T^{-1/2}g|\Bigr) T^{-1/2}g\right\rangle
\geq \varepsilon\|T^{-1/2}g\|^2,
\]
which means that $1+\lambda\|T^{-1/2}g\|^2>0$.

Conversely, we assume that $1+\lambda\|T^{-1/2}g\|^2>0$.
Take an arbitrary $f\in\sH$, and write $f$ as $f=\alpha T^{-1/2}g+h$ with unique $\alpha\in\CC$ and $h\in (\CC T^{-1/2}g)^{\perp}$.
Then
\begin{align*}
\left\langle f, \Bigl(1+\lambda|T^{-1/2}g\rangle\langle T^{-1/2}g|\Bigr) f\right\rangle
&= \|f\|^2+\lambda\alpha^2\|T^{-1/2}g\|^4\\
&= \|\alpha T^{-1/2}g\|^2+\|h\|^2+\lambda\alpha^2\|T^{-1/2}g\|^4\\
&= \alpha^2\|T^{-1/2}g\|^2(1+\lambda\|T^{-1/2}g\|^2)+\|h\|^2,
\end{align*}
whence $1+\lambda|T^{-1/2}g\rangle\langle T^{-1/2}g|$ is non-negative and injective.
This together with the fact that $\lambda|T^{-1/2}g\rangle\langle T^{-1/2}g|$ is of finite rank implies (B5).

The rest of the theorem follows from Theorem \ref{saH}, Theorem \ref{main} and Theorem \ref{gse}.
\end{proof}

\subsection{A Model of  a Harmonic Oscillator Coupled to a Bose Field}\label{ho+bose}
Let $\sH$ be a separable complex Hilbert space.
We consider the following Hamiltonian acting in $L^2(\RR)\otimes\Fb(\sH)$:
\[
H:=\frac{1}{2}\left(p^2+\omega^2 x^2\right)\otimes1 +1\otimes \dGb(T) + \lambda x\otimes\Phi_{\rm S}(g),
\]
where $\lambda\in\RR$, $\omega>0$ are constants, $T$ is an injective non-negative self-adjoint operator acting in $\sH$,
and $g\not=0$ is a non-zero vector in ${\rm dom}(T^{1/2})\cap{\rm dom}(T^{-1/2})$.
We set the domain of $H$ by
\begin{align}
 \dom(H) 
 := \dom\left(p^2\tensor\one\right) 
  \cap \dom\left(x^2\tensor \one\right)
  \cap \dom\Big( \one\tensor \dGb(T)\Big).  \label{dom of arai}
\end{align}
Note that
\begin{align*}
  \dom\left(x\tensor\PhiS(g)\right) \supset \dom(x^2\tensor\one)\cap\dom(\one\tensor\PhiS(g)^2) \supset \dom(H),
\end{align*}
and hence $H$ is well-defined on $\dom(H)$.
The Hamiltonian $H$ was investigated by Arai \cite{Ar81}.

\begin{theorem}\label{ho+boson diag}
Suppose that $|\lambda|<\omega\|T^{-1/2}g\|^{-1}$.
Then $H$ is self-adjoint, and essentially self-adjoint on any core of
\[
\frac{1}{2}\left(p^2+\omega^2 x^2\right)\otimes1 +1\otimes \dGb(T).
\]
Furthermore, there exists a unitary operator 
$U:L^2(\RR,dx)\otimes\Fb(\sH)\to\Fb(\CC\oplus\sH)$ such that 
\[
UHU^*={\rm d}\Gamma_{\rm b}\left(\begin{pmatrix}
\omega^2 & \lambda|1\rangle\langle T^{1/2}g| \\
\lambda |T^{1/2}g\rangle\langle 1| & T^2
\end{pmatrix}^{1/2}\right) + E
\]
with
\[
E=\frac{1}{2}\mathrm{tr}\left(\,\overline{\begin{pmatrix}
\omega^2 & \lambda|1\rangle\langle T^{1/2}g| \\
\lambda |T^{1/2}g\rangle\langle 1| & T^2
\end{pmatrix}^{1/2}
-\begin{pmatrix}
\omega & 0 \\
0 & T
\end{pmatrix}}\,\right)+\frac{\omega}{2}.
\]
\end{theorem}

\begin{proof}
Take an arbitrary $z\in\CC$, and we write it as $z=a+ib$ with $a,b\in\RR$. 
Define a self-adjoint operator $\varphi(z)$ acting in $L^2(\RR)$ by
\[
\varphi(z):=\overline{a\omega^{1/2}x
+b\omega^{-1/2}p}.
\]
We identify $L^2(\RR)$ with $\Fb(\CC)$ via the 
unique unitary operator $u:L^2(\RR)\to\Fb(\CC)$ 
such that
\[
u\varphi(z)u^{-1}=\Phi_{\rm S}(z), \ \ \ \ \ z\in\CC
\]
and that 
\[
u\cdot\left(\frac{\omega}{\pi}\right)^{1/4}
\exp{\left(-\frac{1}{2}\omega x^2\right)}=\Omega.
\]
Then we have
\[
u\left[\frac{1}{2}\left(p^2  +\omega^2x^2\right)\right]u^*  
= \dGb(\omega)+\frac{\omega}{2},
\ \ \ \ \ uxu^*=\omega^{-1/2}\Phi_{\rm S}(1).
\]
Thus
\[
(u\otimes1)H(u\otimes1)^*=\dGb(\omega)\otimes1+1\otimes\dGb(T)
+\omega^{-1/2}\lambda\Phi_{\rm S}(1)
\otimes\Phi_{\rm S}(g)+\frac{\omega}{2}.
\]
We use the natural isomorphism
$\Fb(\CC)\otimes\Fb(\sH)=\Fb(\CC\oplus\sH)$.
Then
\begin{align*}
(u\otimes1)H(u\otimes1)^*&=\dGb(\omega\oplus T) +\omega^{-1/2}\lambda\Phi_{\rm S}(1,0)\Phi_{\rm S}(0,g)+\frac{\omega}{2}\\
&=\dGb(\omega\oplus T) +\frac{\omega^{-1/2}\lambda}{4}\left\{\Phi_{\rm S}(1,g)^2-\Phi_{\rm S}(1,-g)^2\right\}+\frac{\omega}{2}\\
&=:\tilde{H}+\frac{\omega}{2}.
\end{align*}
Note that $\tilde{H}$ is of the form (\ref{hamil}).

We have to check that $\tilde{H}$ satisfies the conditions (B1)--(B6).
However, (B1)--(B4) are obvious.
(B6) follows from Lemma \ref{conjugation lemma}.
Let us prove that $\tilde{H}$ satisfies (B5).
For this, it is sufficient to show that
\begin{align*}
K:=1&+\frac{\omega^{-1/2}\lambda}{2} 
\left|(\omega^{-1/2},T^{-1/2}g)\right\rangle\left\langle (\omega^{-1/2},T^{-1/2}g)\right|\\
&- \frac{\omega^{-1/2}\lambda}{2} \left|(\omega^{-1/2},-T^{-1/2}g)\right\rangle\left\langle (\omega^{-1/2},-T^{-1/2}g)\right|\\
=1&+\omega^{-1}\lambda
\Bigl(\left|(1,0)\right\rangle\left\langle (0,T^{-1/2}g)\right|
+\left|(0,T^{-1/2}g)\right\rangle\left\langle (1,0)\right|\Bigr)
\end{align*}
is an injective non-negative self-adjoint operator because $K-1$ is of finite rank.
Let $\mathscr{K}:=\CC\oplus\CC T^{-1/2}g$.
Then $\mathscr{K}$ reduces $K$, and the restriction of $K$ to $\mathscr{K}$ has the representation matrix
\[
\begin{pmatrix}
1 & \omega^{-1}\lambda\|T^{-1/2}g\| \\
\omega^{-1}\lambda\|T^{-1/2}g\| & 1
\end{pmatrix}
\]
with respect to the orthonormal basis $\left\{(1,0),\ (0,T^{-1/2}g/\|T^{-1/2}g\|)\right\}$.
All of its eigenvalues are positive if and only if $1-\omega^{-2}\lambda^2\|T^{-1/2}g\|^2>0$.
On the other hand, the restriction of $K$ to $\mathscr{K}^{\perp}=\{0\}\oplus(\CC T^{-1/2}g)^{\perp}$ is the identity.
Thus the kernel of $K$ is trivial.
We conclude that $K$ is non-negative and injective, and hence $\tilde{H}$ satisfies (B5).

Therefore $\tilde{H}$ satisfies the conditions (B1)--(B6).
The theorem now follows from Theorem \ref{saH}, Theorem \ref{main} and Theorem \ref{gse}.
\end{proof}

\subsection{The Pauli-Fierz Model with $x^2$-potentials in the Dipole Approximation}\label{PFmodel}
Let $\sH$ be a separable complex Hilbert space.
We consider the following Hamiltonian $H$ acting in $L^2(\RR^d)\otimes\Fb(\sH)$:
\[
H:=\frac{1}{2}\sum_{j=1}^{d}\Bigl(p_j\otimes 1+1\otimes\Phi_{\rm S}(g_j)\Bigr)^2 
+\frac{1}{2}\sum_{j=1}^d \omega_j^2x_j^2\otimes 1 +1\otimes\dGb(T),
\]
where $\omega_j>0\ (j=1,\cdots, d)$ are constants, $T$ is an injective non-negative self-adjoint operator acting in $\sH$,
and $g_1,\cdots,g_d$ are vectors in ${\rm dom}(T^{1/2})\cap{\rm dom}(T^{-1/2})$.
We set the domain of $H$ by
\begin{align}
 \dom(H) 
 := \dom\left(\sum_{j=1}^d p_j^2\tensor\one\right) 
  \cap \dom\left(\sum_{j=1}^d x_j^2\tensor \one\right)
  \cap \dom\Big( \one\tensor \dGb(T)\Big).  \label{dom of Hpf}
\end{align}
Note that
\begin{align*}
  \dom\left((p_j\tensor 1)(1\otimes \PhiS(g_j))\right) \supset \dom(p_j^2\tensor\one)\cap\dom(\one\tensor\PhiS(g_j)^2) \supset \dom(H),
\end{align*}
and hence $H$ is well-defined on $\dom(H)$.
In particular, we have
\begin{align}
H &= \frac{1}{2}\sum_{j=1}^d(p_j^2+\omega_j^2x_j^2)\otimes 1 + 1\otimes\dGb(T)\notag\\
&~~~~+ \frac{1}{2}\sum_{j=1}^d1\otimes\PhiS(g_j)^2 + \sum_{j=1}^dp_j\otimes\PhiS(g_j)\label{gpfh}
\end{align}
as an operator equality.
This Hamiltonian is an abstract version of the Pauli-Fierz Hamiltonian with $x^2$-potentials in the dipole approximation studied by Arai \cite{A83}.

\begin{theorem}
Suppose that there exists a conjugation $J$ on $\sH$ such that 
\[
JTJ=T,\ \ \ \ \ Jg_j=g_j,\ \ \ \ \ j=1,\cdots,d.
\]
Then $H$ is self-adjoint, and essentially self-adjoint on any core of
\[
\frac{1}{2}\sum_{j=1}^d(p_j^2+\omega_j^2x_j^2)\otimes 1 + 1\otimes\dGb(T).
\]
Furthermore, there exists a unitary operator $U:L^2(\RR^d)\otimes\Fb(\sH)\to\Fb(\CC^d\oplus\sH)$ such that 
\[
UHU^*=\dGb\left(\sqrt{{\rm diag}(\omega_1^2,\cdots,\omega_d^2)\oplus T^2 + W}\right)
+E
\]
with
\[
E=\frac{1}{2}\mathrm{tr}\left(\,\overline{\sqrt{{\rm diag}(\omega_1^2,\cdots,\omega_d^2)\oplus T^2 + W}-{\rm diag}(\omega_1,\cdots,\omega_d)\oplus T}\,\right)+\sum_{j=1}^d\frac{\omega_j}{2},
\]
where $W$ is a finite rank operator on $\CC^d\oplus\sH$ defined by
\[
W:=\sum_{j=1}^d
\begin{pmatrix}
0 & \omega_j|e_j\rangle\langle T^{1/2}g_j| \\
\omega_j |T^{1/2}g_j\rangle\langle e_j| &  |T^{1/2}g_j\rangle\langle T^{1/2}g_j|
\end{pmatrix},
\]
and $\{e_j\}_{j=1}^d$ is the standard basis of $\CC^d$.
\end{theorem}

\begin{proof}
Take an arbitrary $z\in\CC^d$, and we write it as $z=\sum_{j=1}^{d}(a_j+ib_j)e_j$ with $a_j, b_j\in\RR\ (j=1,\cdots,d)$.
Define a self-adjoint operator $\varphi(z)$ acting in $L^2(\RR^d)$ by
\[
\varphi(z):=\overline{\sum_{j=1}^d\left(a_j\omega_j^{1/2}x_j
+b_j\omega_j^{-1/2}p_j\right)},\ \ \ \ \ z\in\CC^d.
\]
Let $u_1:L^2(\RR^d)\to\Fb(\CC^d)$ be a unique unitary operator 
such that
\[
u_1\varphi(z)u_1^{-1}=\Phi_{\rm S}(z), \ \ \ \ \ z\in\CC^d
\]
and that 
\[
u_1\cdot\left(\prod_{j=1}^d\frac{\omega_j}{\pi}\right)^{1/4}
\exp{\left(-\frac{1}{2}\sum_{j=1}^d\omega_j x_j^2\right)}=\Omega.
\]
Set 
\[
\omega:=\sum_{j=1}^d\omega_j|e_j\rangle\langle e_j|={\rm diag}(\omega_1,\cdots,\omega_d),\ \ \ \ \ 
c:=\sum_{j=1}^d\frac{\omega_j}{2}.
\]
It follows that
\[
u_1\left[\frac{1}{2}\sum_{j=1}^d\left(p_j^2  + \omega_j^2 x_j^2\right)\right]u_1^*
= \dGb(\omega)+c,
\ \ \ \ \ u_1p_ju_1^*=\omega_j^{1/2}\Phi_{\rm S}(ie_j)
\]
for all $j=1,\cdots,d$.
Letting $u_2:=\Gamma_{\rm b}(-i):=\oplus_{n=0}^\infty\otimes^n(-i)$ with $\otimes^0(-i):=1$, we have
\[
u_2\dGb(\omega)u_2^*=\dGb(\omega),\qquad u_2\Phi_{\rm S}(ie_j)u_2^*=\Phi_{\rm S}(e_j),
\]
and hence
\begin{align*}
\tilde{H} &:= (u_2u_1\otimes 1)H(u_2u_1\otimes 1)^*-c\\
&\phantom{:} =\dGb(\omega)\otimes 1 + 1\otimes\dGb(T)+ \frac{1}{2}\sum_{j=1}^d1\otimes\PhiS(g_j)^2 + \sum_{j=1}^d\omega_j^{1/2}\Phi_{\rm S}(e_j)\otimes\PhiS(g_j)
\end{align*}
holds.
We use the natural isomorphism
$\Fb(\CC^d)\otimes\Fb(\sH)=\Fb(\CC^d\oplus\sH)$.
Then we get
\begin{align*}
\tilde{H}&=\dGb(\omega\oplus T)
+\frac{1}{2}\sum_{j=1}^d\Phi_{\rm S}(0,g_j)^2+\sum_{j=1}^d\omega_j^{1/2}\Phi_{\rm S}(e_j,0)\Phi_{\rm S}(0,g_j)\\
&=\dGb(\omega\oplus T)
+\frac{1}{2}\sum_{j=1}^d\Phi_{\rm S}(0,g_j)^2
+\frac{1}{4}\sum_{j=1}^d\omega_j^{1/2}\left[\Phi_{\rm S}(e_j,g_j)^2-\Phi_{\rm S}(e_j,-g_j)^2\right]
\end{align*}
as an operator equality.
Note that $\tilde{H}$ is of the form (\ref{hamil}).

Let us prove that $\tilde{H}$ satisfies the condition (B5).
For this, it is sufficient to show that
\begin{align*}
K:=1&+\sum_{j=1}^d\left|(0,T^{-1/2}g_j)\right\rangle\left\langle (0,T^{-1/2}g_j)\right|\\
&+ \sum_{j=1}^d\frac{\omega_j^{1/2}}{2} 
\left|(\omega^{-1/2}e_j,T^{-1/2}g_j)\right\rangle\left\langle (\omega^{-1/2}e_j,T^{-1/2}g_j)\right|\\
&+ \sum_{j=1}^d\left(-\frac{\omega_j^{1/2}}{2}\right) 
\left|(\omega^{-1/2}e_j,-T^{-1/2}g_j)\right\rangle\left\langle (\omega^{-1/2}e_j,-T^{-1/2}g_j)\right|\\
=1&+ \sum_{j=1}^d\left|(0,T^{-1/2}g_j)\right\rangle\left\langle (0,T^{-1/2}g_j)\right|\\
&+\sum_{j=1}^d
\Bigl(\left|(e_j,0)\right\rangle\left\langle (0,T^{-1/2}g_j)\right|
+\left|(0,T^{-1/2}g_j)\right\rangle\left\langle (e_j,0)\right|\Bigr)
\end{align*}
is an injective non-negative self-adjoint operator because $K-1$ is of finite rank.
Since $1=\sum_{j=1}^d\big(|e_j\rangle\langle e_j|\oplus (1/d)\big)$, 
we can write $K$ as $K=\sum_{j=1}^{d}K_j$, where
\begin{align*}
K_j&:=|e_j\rangle\langle e_j|\oplus \frac{1}{d}+\left|(0,T^{-1/2}g_j)\right\rangle\left\langle(0,T^{-1/2}g_j)\right|\\
&\phantom{:=}\ +\left|(e_j,0)\right\rangle\left\langle (0,T^{-1/2}g_j)\right|
+\left|(0,T^{-1/2}g_j)\right\rangle\left\langle (e_j,0)\right|
\end{align*}
is a self-adjoint operator for each $j=1,\cdots,d$.

Fix an arbitrary $j\in\{1,\cdots,d\}$, and let $\mathscr{L}_j$ be the complex linear subspace of $\CC^d$ spanned by vectors $e_1,\cdots,e_{j-1},e_{j+1},\cdots,e_d$.
We show that $K_j$ is a non-negative self-adjoint operator whose kernel is equal to $\mathscr{L}_j\oplus\{0\}$.
The case $g_j=0$ is trivial, and so we may assume that $g_j\not=0$.
Let $\mathscr{K}_j:=\CC e_j\oplus\CC T^{-1/2}g_j$.
Then $\mathscr{K}_j$ reduces $K_j$, and the restriction of $K_j$ to $\mathscr{K}_j$ has the representation matrix
\[
\begin{pmatrix}
1 & \|T^{-1/2}g_j\| \\
\|T^{-1/2}g_j\| & 1/d+\|T^{-1/2}g_j\|^2
\end{pmatrix}
\]
with respect to the orthonormal basis $\left\{(e_j,0),\ (0,T^{-1/2}g_j/\|T^{-1/2}g_j\|)\right\}$.
A straightforward computation shows that all of its eigenvalues are positive.
On the other hand, the restriction of $K_j$ to $\mathscr{K}_j^{\perp}=\mathscr{L}_j\oplus(\CC T^{-1/2}g_j)^{\perp}$ is $0\oplus(1/d)$.
Therefore the kernel of $K_j$ is $\mathscr{L}_j\oplus\{0\}$.

Since $K=\sum_{j=1}^dK_j$ and $K_1,\cdots,K_d$ are all non-negative,
we conclude that $K$ is non-negative and injective, and thus $\tilde{H}$ satisfies (B5).

Therefore $\tilde{H}$ satisfies the conditions (B1)--(B5).
We define a conjugation operator $\tilde{J}$ on $\CC^d\oplus\sH$ by
\[
\tilde{J}((z_1,\cdots,z_d)\oplus f):=(\bar{z_1},\cdots,\bar{z_d})\oplus Jf,\qquad (z_1,\cdots,z_d)\in\CC^d,\ f\in\sH.
\]
Then (B6) holds.
By Theorem \ref{saH}, $\tilde{H}$ is self-adjoint, and hence so is $H$. 
The rest of the theorem follows from Theorem \ref{main} and Theorem \ref{gse}.
This completes the proof.
\end{proof}

\subsection{The Translation Invariant Pauli-Fierz Model in the Dipole Approximation}\label{translation PF}
Let $\sH$ be a separable complex Hilbert space.
We consider the following Hamiltonian acting in $L^2(\RR^d)\otimes\Fb(\sH)$: 
\[
H:=\frac{1}{2}\sum_{j=1}^{d}\Bigl(p_j\otimes 1+1\otimes\Phi_{\rm S}(g_j)\Bigr)^2 +1\otimes\dGb(T),
\]
where $T$ is an injective non-negative self-adjoint operator acting in $\sH$,
and $g_1,\cdots,g_d$ are vectors in ${\rm dom}(T^{1/2})\cap{\rm dom}(T^{-1/2})$.
We set the domain of $H$ by
\begin{equation}
 \dom(H) 
 := \dom\left(\sum_{j=1}^d p_j^2\tensor\one\right) \cap \dom\Big( \one\tensor \dGb(T)\Big).  
\end{equation}
Similar to Subsection \ref{PFmodel}, $H$ is well-defined on $\dom(H)$.
This Hamiltonian is an abstract version of the translation invariant Pauli-Fierz Hamiltonian in the dipole approximation studied by Arai \cite{A83b}.

We suppose that there exists a conjugation $J$ on $\sH$ such that 
\[
JTJ=T,\ \ \ \ \ Jg_j=g_j,\ \ \ \ \ j=1,\cdots,d.
\]
Let $\mathscr{F}_d:L^2(\RR^d,dx)\to L^2(\RR^d,dP)$ be the Fourier transform defined by
\[
(\mathscr{F}_df)(P) := \frac{1}{(2\pi)^{d/2}}\int_{\RR^d}f(x)e^{-iPx}\,dx,\ \ \ \ \ f\in L^2(\RR^d,dx),\ P\in\RR^d
\]
in the $L^2$-sense, where $Px:=\sum_{j=1}^dP_jx_j$.
We use the natural isomorphism
\[
L^2(\RR^d,dP)\otimes\Fb(\sH) = L^2\left(\RR^d,dP;\Fb(\sH)\right) = \int_{\RR^d}^\oplus\Fb(\sH)\,dP.
\]
For the details, see e.g., \cite[Section 2.7, 2.8 and 3.11]{A18}.
Then we have the operator equality
\[
(\mathscr{F}_d\otimes1)H(\mathscr{F}_d\otimes1)^*
= \int_{\RR^d}^{\oplus}H(P)\,dP,
\]
where
\begin{align*}
H(P)&:=\frac{1}{2}\sum_{j=1}^{d}\Bigl(P_j+\Phi_{\rm S}(g_j)\Bigr)^2 +\dGb(T)\\
&\phantom{:}=\dGb(T)+\frac{1}{2}\sum_{j=1}^{d}\Phi_{\rm S}(g_j)^2
+\sum_{j=1}^{d}P_j\Phi_{\rm S}(g_j)+\sum_{j=1}^{d}\frac{P_j^2}{2}
\end{align*}
for $P=(P_1,\cdots,P_d)\in\RR^d$.
As we will see below, $H(P)$ is self-adjoint on $\dom(\dGb(T))$.
The main purpose of this subsection is to investigate $H(P)$.

It follows from Theorem \ref{saH} that 
\[
  H_0 := \dGb(T) + \frac{1}{2}\sum_{j=1}^{d}\Phi_{\rm S}(g_j)^2
\]
is self-adjoint on $\dom(\dGb(T))$.
By Theorem \ref{main} and Theorem \ref{gse}, there exists a unitary operator $U$ on $\Fb(\sH)$,
independent of $P\in\RR^d$, such that
\[
  UH_0U^* = \dGb(S) + E, \ \ \ \ \ S:=\sqrt{T^2+\sum_{j=1}^d|T^{1/2}g_j\rangle\langle T^{1/2}g_j|}
\]
with $E:=\mathrm{tr}(\overline{S-T})/2$.
Since $U\Phi_{\rm S}(g_j)U^*=\Phi_{\rm S}(S^{-1/2}T^{1/2}g_j)$, we get the operator equality
\begin{equation}\label{UH(P)U^*}
UH(P)U^{*} = \dGb(S) +\Phi_{\rm S}\left(S^{-1/2}T^{1/2}\sum_{j=1}^{d} P_jg_j\right)+\sum_{j=1}^{d}\frac{P_j^2}{2} +E.
\end{equation}
The right-hand side is so called a van Hove Hamiltonian, which was studied in 
\cite[Chapter 12]{A00}, \cite[Chapter 13]{A18}, \cite{De03} and \cite[Section 11.6]{DG13}.
By Lemma \ref{domain}, the vector $S^{-1/2}T^{1/2}g_j$ is in $\dom(S^{-1/2})$.
It follows from the Kato-Rellich theorem that $UH(P)U^{*}$ is self-adjoint on $\dom(\dGb(S))$ 
and bounded from below (see \cite[Theorem 13.1]{A18} or \cite[Proposition 3.13]{De03}). 
In particular, $H(P)$ is self-adjoint on $\dom(\dGb(T))$.

We next study the existence/absence of a ground state of $H(P)$.
By \cite[Theorem 13.5]{A18}, the lowest energy value $E(P)$ of $H(P)$, which is the infimum of the spectrum of $H(P)$, is given by
\[
E(P) = -\frac{1}{2} \left\|S^{-1}T^{1/2}\sum_{j=1}^dP_jg_j\right\|^2 + \sum_{j=1}^d\frac{P_j^2}{2} + E.
\]
It is, however, not obvious whether $H(P)$ has a ground state or not.
By \eqref{UH(P)U^*}, $H(0)$ has a ground state, and thus $E(0)=E$ is the ground state energy of $H(0)$.
For $P\not=0$, the existence/absence of a ground state of $H(P)$ corresponds to the infrared regularity/singularity condition.

\begin{thm}\label{e/a gs}
The following are equivalent:
\begin{itemize}{}{}
\item[(1)] $H(P)$ has a ground state for all $P\in\RR^d$.
\item[(2)] $g_j\in\dom(T^{-1})$ for all $j=1,\cdots,d$.
\end{itemize}
\end{thm}

\begin{proof}
It follows from \cite[Proposition 3.13]{De03} or \cite[Theorem 11.73 (3)]{DG13} that for each $P\in\RR^d$, $UH(P)U^{*}$ has a ground state if only if 
\begin{equation}\label{gs condition}
S^{-1/2}T^{1/2}\sum_{j=1}^{d} P_jg_j\in\dom(S^{-1}).
\end{equation}
We first show (2) $\Rightarrow$ (1).
By Lemma \ref{domain-3/2}, we have $\dom(T^{-3/2})=\dom(S^{-3/2})$.
This together with assumption (2) implies that
\[
T^{1/2}\sum_{j=1}^{d} P_jg_j\in\dom(T^{-3/2})=\dom(S^{-3/2}),
\]
which is equivalent to \eqref{gs condition}, and thus $H(P)$ has a ground state for all $P\in\RR^d$.

We next show (1) $\Rightarrow$ (2).
Set
\[
\sD :=\dom(T^2)\cap\dom(T^{-1/2})=\dom(S^2)\cap\dom(S^{-1/2}),
\]
and take an arbitrary $u\in\sD $.
Then we have
\begin{align}\label{S^3/2 equality}
&S^{3/2}u = S^2S^{-1/2}u = \left(T^2+\sum_{j=1}^d\ket{T^{1/2}g_j}\bra{T^{1/2}g_j}\right)S^{-1/2}u\notag\\
&= T^{3/2}\cdot\overline{T^{1/2}S^{-1/2}}u+\left(\sum_{j=1}^d\ket{T^{1/2}g_j}\bra{S^{-1/2}T^{1/2}g_j}\right)u.
\end{align}
Since $\sD $ is a core of $S^{3/2}$, we obtain 
\[
\dom(S^{3/2})\subset\dom\left(T^{3/2}\cdot\overline{T^{1/2}S^{-1/2}}\right)
\]
and \eqref{S^3/2 equality} for all $u\in\dom(S^{3/2})$.

On the other hand, for any $u\in\dom(S)\cap\dom(S^{-1})$, we have
\begin{align*}
\left(\overline{TS^{-1}}\right)^*\left(\overline{TS^{-1}}\right)u 
&= S^{-1}T^2S^{-1}u = S^{-1}\left(S^{2}-\sum_{j=1}^d\ket{T^{1/2}g_j}\bra{T^{1/2}g_j}\right)S^{-1}u\\
&= \left(1-\sum_{j=1}^d\ket{S^{-1}T^{1/2}g_j}\bra{S^{-1}T^{1/2}g_j}\right)u,
\end{align*}
and thus we get the operator equality
\begin{equation*}
\left(\overline{TS^{-1}}\right)^*\left(\overline{TS^{-1}}\right) = 1-\sum_{j=1}^d\ket{S^{-1}T^{1/2}g_j}\bra{S^{-1}T^{1/2}g_j}=:A.
\end{equation*}
Since $\overline{TS^{-1}}$ is bijective by Lemma \ref{domain}, so is $A$.
Let $\mathscr{E}$ be the subspace spanned by $S^{-1}T^{1/2}g_1,\cdots,S^{-1}T^{1/2}g_d$.
Then $A$ maps $\mathscr{E}$ into $\mathscr{E}$.
The bijectivity of $A$ implies that the restriction $A|_{\mathscr{E}}$ of $A$ onto $\mathscr{E}$ is injective.
Since $\mathscr{E}$ is finite dimensional, $A|_{\mathscr{E}}$ is bijective.
For each $\ell=1,\cdots,d$, we choose a vector $u_\ell\in\mathscr{E}$ so that $S^{-1}T^{1/2}g_\ell=Au_\ell$.
We now use the assumption (1), which means that $\mathscr{E}$ is contained in $\dom(S^{-1/2})$. 
In particular, each $u_\ell$ is in $\dom(S^{-1/2})$, and hence $S^{-1/2}u_\ell\in\dom(S^{3/2})$.
Letting $u=S^{-1/2}u_\ell$ in \eqref{S^3/2 equality}, we have
\begin{align*}
&T^{3/2}\cdot\overline{T^{1/2}S^{-1/2}}\cdot S^{-1/2}u_\ell
= \left(S^{3/2}-\sum_{j=1}^d\ket{T^{1/2}g_j}\bra{S^{-1/2}T^{1/2}g_j}\right)S^{-1/2}u_\ell\\
&=S\left(1-\sum_{j=1}^d\ket{S^{-1}T^{1/2}g_j}\bra{S^{-1}T^{1/2}g_j}\right)u_\ell
=SAu_\ell =T^{1/2}g_\ell.
\end{align*}
The left-hand side is in the range of $T^{3/2}$, and thus $g_\ell\in\dom(T^{-1})$.
This finishes the proof.
\end{proof}

In a concrete setting (see e.g., \cite{A83b}), it follows from the property of polarization vectors (see e.g., \cite[equality (11.26)]{A18}) 
that the set $\{T^{-1/2}g_j\}_{j=1}^d$ satisfies
\begin{equation}\label{-1/2g_jONS}
\langle T^{-1/2}g_j,T^{-1/2}g_\ell\rangle =\|T^{-1/2}g_1\|^2\delta_{j\ell},\qquad j,\ell=1,\cdots,d
\end{equation}
where $\delta_{j\ell}$ denotes the Kronecker delta.
In our setting, if we further suppose \eqref{-1/2g_jONS}, then we get a stronger result than Theorem \ref{e/a gs}, which is an abstract version of \cite[Theorem 3.3]{A83b}.

\begin{thm}
Suppose \eqref{-1/2g_jONS}.
Let $P\in\RR^d$ be arbitrary.
Then the lowest energy value $E(P)$ of $H(P)$ is computed as
\[
E(P) = \frac{1}{2(1+\|T^{-1/2}g_1\|^2)}\sum_{j=1}^dP_j^2 + E.
\]
Furthermore the following are equivalent:
\begin{itemize}{}{}
\item[(1)] $H(P)$ has a ground state.
\item[(2)] $\sum_{j=1}^d P_jg_j\in\dom(T^{-1})$.
\end{itemize}
\end{thm}

\begin{proof}
Let $A:=1+\sum_{j=1}^d\ket{T^{-1/2}g_j}\bra{T^{-1/2}g_j}$.
It follows from \eqref{-1/2g_jONS} that
\begin{equation}\label{A^{-1}}
A^{-1} = 1-\frac{1}{1+\|T^{-1/2}g_1\|^2}\sum_{j=1}^d\ket{T^{-1/2}g_j}\bra{T^{-1/2}g_j}.
\end{equation}

Since $S^2 = TAT$, we have the operator equality $S^{-2}=T^{-1}A^{-1}T^{-1}$.
Since $\dom(S^{-2})$ is a core of $S^{-1}$, we obtain
\begin{equation}\label{S^{-1}}
\langle S^{-1}u,S^{-1}v\rangle = \langle A^{-1/2}T^{-1}u,A^{-1/2}T^{-1}v\rangle, \qquad u,v\in \dom(S^{-1}).
\end{equation}

Let $g:=\sum_{j=1}^dP_jg_j$.
For any $u\in\dom(T^{1/2})\cap\dom(T^{-1/2})$, it holds that
\begin{align}\label{key equation}
&(1+\|T^{-1/2}g_1\|^2)\langle T^{-1/2}u,T^{-1/2}g\rangle 
= \langle T^{-1/2}u,A^{-1}T^{-1/2}g\rangle\notag\\
&= \langle A^{-1/2}T^{-1}\cdot T^{1/2}u,A^{-1/2}T^{-1}\cdot T^{1/2}g\rangle
= \langle S^{-1}T^{1/2}u,S^{-1}T^{1/2}g\rangle\notag\\
&= \langle S^{-1/2}\cdot S^{-1/2}T^{1/2}u,S^{-1}T^{1/2}g\rangle,
\end{align}
where we have used \eqref{A^{-1}} at the second equality and \eqref{S^{-1}} at the third equality.

Recall that $H(P)$ has a ground state if and only if $T^{1/2}g\in\dom(S^{-3/2})$.
We first show (1) $\Rightarrow$ (2).
By \eqref{key equation}, we have
\begin{equation}\label{(1)implies(2)}
\langle T^{-1/2}u,T^{-1/2}g\rangle = \frac{1}{1+\|T^{-1/2}g_1\|^2}\langle u,\left(S^{-1/2}T^{1/2}\right)^*S^{-3/2}T^{1/2}g\rangle
\end{equation}
for all $u\in\dom(T^{1/2})\cap\dom(T^{-1/2})$.
Since $\dom(T^{1/2})\cap\dom(T^{-1/2})$ is a core of $T^{-1/2}$, we get \eqref{(1)implies(2)} for all $u\in\dom(T^{-1/2})$, 
and thus $g\in\dom(T^{-1})$.

We next show (2) $\Rightarrow$ (1).
By Lemma \ref{domain}, for any $v\in\dom(T^{1/2})\cap\dom(T^{-1/2})$, we obtain $T^{-1/2}S^{1/2}v\in\dom(T^{1/2})\cap\dom(T^{-1/2})$.
Letting $u=T^{-1/2}S^{1/2}v$ in \eqref{key equation}, we have
\begin{align}\label{(2)implies(1)}
\langle S^{-1/2}v,S^{-1}T^{1/2}g\rangle 
&= (1+\|T^{-1/2}g_1\|^2)\langle T^{-1/2}S^{1/2}v,T^{-1}g\rangle\notag\\ 
&= (1+\|T^{-1/2}g_1\|^2)\langle v,\left(T^{-1/2}S^{1/2}\right)^*T^{-1}g\rangle.
\end{align}
Since $\dom(T^{1/2})\cap\dom(T^{-1/2})$ is a core of $S^{-1/2}$, we get \eqref{(2)implies(1)} for all $v\in\dom(S^{-1/2})$.
Therefore $T^{1/2}g\in\dom(S^{-3/2})$, which implies that $H(P)$ has a ground state.

Finally, we show the expression of the lowest energy value.
Note that, for this proof, we do not assume (1) or (2).
Letting $u=v=T^{1/2}g$ in \eqref{S^{-1}}, one has
\begin{align*}
E(P) &= -\frac{1}{2}\left\|S^{-1}T^{1/2}g\right\|^2+\sum_{j=1}^d\frac{P_j^2}{2}+ E\\
&= -\frac{1}{2(1+\|T^{-1/2}g_1\|^2)}\|T^{-1/2}g\|^2+\sum_{j=1}^d\frac{P_j^2}{2}+ E\\
&= \frac{1}{2(1+\|T^{-1/2}g_1\|^2)}\sum_{j=1}^d P_j^2+ E.
\end{align*}
This completes the proof.
\end{proof}



\appendix

\section{Inequalities on Creation-Annihilation Operators and Second Quantizations}
Let $(M,\mu)$ be a measure space. 
Suppose that $L^2(M):=L^2(M,d\mu)$ is separable.
The space $\tensor_\mathrm{s}^nL^2(M)$ can be identified with the set of square integrable symmetric functions.
\begin{align*}
L^2_\mathrm{sym}(M^n)
:= \big\{\Psi \in L^2(M^n) ~\big|~ \Psi(k_1,\cdots,k_n)=\Psi(k_{\sigma(1)},\cdots,k_{\sigma(n)}), \sigma\in S_n\big\}.
\end{align*}
Let us consider the Cartesian product space
\begin{align*}
  \sF^\mathsf{x} := \mathop{\mathsf{X}}_{n=0}^\infty L^2_\mathrm{sym}(M^n),
\end{align*}
where $\CC:=L^2_\mathrm{sym}(M^0)$.
Then the Fock space $\Fb(L^2(M))$ can be identified with a subset of $\sF^\mathsf{x}$.
For $\Psi=(\Psi^{(n)})_{n=0}^\infty \in \sF^\mathsf{x}$, we define an informal norm by
\begin{align*}
  \norm{\Psi}^2 := \sum_{n=0}^\infty \norm{\Psi^{(n)}}_{L^2(M^n)}^2 \in [0,+\infty].
\end{align*}
An inner product of $\Psi,\Phi \in \sF^\mathsf{x}$ is defined by
\begin{align*}
  \inner{\Psi}{\Phi} := \sum_{n=0}^\infty \inner{\Psi^{(n)}}{\Phi^{(n)}}
\end{align*}
if the sum converges.
For $\Psi = (\Psi^{(n)})_{n=0}^\infty \in \sF^\mathsf{x}$ and $k\in M$, we define
$A(k)\Psi \in \sF^\mathsf{x}$ by
\begin{align}
 (A(k)\Psi)^{(n)}(\cdot) := \sqrt{n+1}\Psi^{(n+1)}(k,\cdot) \in L^2_\mathrm{sym}(M^n),
 \qquad n=0,1,2,\cdots.  \label{defAk}
\end{align}
Note that $A(k)\Psi$ is defined for $\mu$-a.e.\,$k\in M$.
For finite particle state $\Phi \in \Fbz$, $\inner{\Phi}{A(k)\Psi}$ consists of 
finite sum and 
\begin{align*}
  \inner{\Phi}{A(f)\Psi} = \int_Md\mu(k)\,\overline{f(k)}\inner{\Phi}{A(k)\Psi},
 \qquad \Psi\in\dom(A(f))
\end{align*}
holds.
\begin{lem}{\label{ak1}}
 Let $Q(k)>0$ be a measurable function. 
The multiplication operator by $Q(k)$ acting in $L^2(M)$ is also denoted by $Q$.
Then, for all $\Psi\in\sF^\mathsf{x}$, 
$\Psi\in\dom(\dGb(Q)^{1/2})$ if and only if
\begin{align*}
  \int_M Q(k)\norm{A(k)\Psi}^2 \,d\mu(k) < \infty.
\end{align*}
In this case, the equality
\begin{align*}
  \norm{\dGb(Q)^{1/2}\Psi}^2 = \int_M Q(k)\norm{A(k)\Psi}^2 \,d\mu(k) 
\end{align*}
holds. Moreover, if $f\in\dom(Q^{-1/2})$, $\Psi\in\dom(A(f))$ and 
\begin{align*}
  \inner{\Phi}{A(f)\Psi} = \int_M \overline{f(k)}\inner{\Phi}{A(k)\Psi} \,d\mu(k)
\end{align*}
hold for $\Psi\in\dom(\dGb(Q)^{1/2})$ and $\Phi\in\Fb(L^2(M))$.
\end{lem}
\begin{proof}
 The lemma directly follows from the definitions of second quantization operator and $A(k)$.
\end{proof}

\begin{lem}{\label{ak3}}
Let $T$ be an injective non-negative self-adjoint operator acting in a separable Hilbert space and $f_1,\cdots,f_n\in \dom(T^{-1/2})$.
Then
\begin{align*}
  \dom(\dGb(T)^{n/2})\subset \dom(A(f_1)\cdots A(f_n)),
\end{align*}
and the bound 
\begin{align}
  \norm{A(f_1)\cdots A(f_n)\Psi}
 \leq \norm{T^{-1/2}f_1}\cdots \norm{T^{-1/2}f_n} \cdot \norm{\dGb(T)^{n/2}\Psi} \label{ak3-1}
\end{align}
holds for $\Psi\in\dom(\dGb(T)^{n/2})$.
In the case of $n=2$, the bound
\begin{align}
  \norm{A(f_1) A(f_2)\Psi}
 \leq \norm{T^{-1/2}f_1}\cdot \norm{T^{-1/2}f_2}
  \big( \norm{\dGb(T)\Psi}^2 - \norm{\dGb(T^2)^{1/2}\Psi}^2 \big)^{1/2} \label{ak3-2}
\end{align}
holds.
\end{lem}
\begin{proof}
Since any self-adjoint operator $T$ is unitarily equivalent to a multiplication operator on an $L^2$-space,
it is enough to prove the lemma in the case of $\sH=L^2(M)$, $T=Q$.
By Lemma \ref{ak1}, for $\Psi\in\dom(\dGb(T)^{n/2})$ and $\Phi\in\Fbz$, we have
\begin{align*}
 \inner{A^*(f_1)\cdots A^*(f_n)\Phi}{\Psi}
& =  \int_{M^n}  \overline{f_1(k_1)}\cdots \overline{f_n(k_n)} \inner{\Phi}{A(k_1)\cdots A(k_n)\Psi}\,d\mu(k_1)\cdots d\mu(k_n).
\end{align*}
We set $\norm{\Phi}=1$, $f_j:=f_j(k_j)$, $Q_j^{1/2}:=Q(k_j)^{1/2}$ and $d\mu:=d\mu(k_1)\cdots d\mu(k_n)$.  
Then we have
\begin{align*}
& |\inner{A^*(f_1)\cdots A^*(f_n)\Phi}{\Psi}| \\
& \leq  \int_{M^n}  |f_1 \cdots f_n|  \cdot \norm{A(k_1)\cdots A(k_n)\Psi}\,d\mu \\
& \leq  \bigg( \int_{M^n}  \prod_{j=1}^n |Q_j^{-1/2}f_j|^2\,d\mu \bigg)^{1/2}  
  \bigg( \int_{M^n} Q_1\cdots Q_n \norm{A(k_1)\cdots A(k_n)\Psi}^2\,d\mu \bigg)^{1/2} \\
& \leq  \norm{Q^{-1/2}f_1} \cdots \norm{Q^{-1/2}f_n}
  \bigg( \int_{M^n} Q_1\cdots Q_n \norm{A(k_1)\cdots A(k_n)\Psi}^2\,d\mu \bigg)^{1/2}.
\end{align*}
By the definition of $A(k)$, one has
\begin{align*}
&  \norm{A(k_1)\cdots A(k_n)\Psi}^2 \\
& = \sum_{N=0}^\infty (N+1)\norm{(A(k_2)\cdots A(k_n)\Psi)^{(N+1)}(k_1,\cdot)}^2 \\
& = \sum_{N=0}^\infty (N+1)(N+2)\norm{(A(k_3)\cdots A(k_n)\Psi)^{(N+2)}(k_1,k_2,\cdot)}^2 \\
& = \sum_{N=0}^\infty (N+1)\cdots (N+n) \norm{\Psi^{(N+n)}(k_1,\cdots,k_n,\cdot)}^2 \\
& = \sum_{N=0}^\infty (N+1)\cdots (N+n) \int_{M^N} d\mu(k_{n+1})\cdots d\mu(k_{n+N})\,|\Psi^{(N+n)}(k_1,\cdots,k_{N+n})|^2.
\end{align*}
Therefore, we have
\begin{align}
&\int_{M^n} Q_1\cdots Q_n \norm{A(k_1)\cdots A(k_n)\Psi}^2\,d\mu  \notag \\
&= \sum_{N=0}^\infty  \int_{M^{N+n}} d\mu(k_1)\cdots d\mu(k_{n+N})\,\frac{(N+n)!}{N!} Q_1\cdots Q_n| \Psi^{(N+n)}(k_1,\cdots,k_{N+n})|^2 \notag \\
&= \sum_{N=n}^\infty  \int_{M^N} d\mu(k_1)\cdots d\mu(k_N)\,\frac{N!}{(N-n)!} Q_1\cdots Q_n| \Psi^{(N)}(k_1,\cdots,k_N)|^2 \notag \\
&= \sum_{N=n}^\infty  \int_{M^N} d\mu(k_1)\cdots d\mu(k_N) \mathop{\sum_{j_1,\cdots,j_n=1}^N}_{\sharp \{j_1,\cdots,j_n\}=n} 
   Q_{j_1}\cdots Q_{j_n}| \Psi^{(N)}(k_1,\cdots,k_N)|^2,  \label{estxx1}
\end{align}
where, in the last step, we used the symmetry of $\Psi^{(N)}$.
In the case $n=2$, we have 
\begin{align}
\eqref{estxx1}
& = \sum_{N=2}^\infty \int_{M^N} d\mu(k_1)\cdots d\mu(k_N)\,\bigg\{ \Big( \sum_{j=1}^N Q_j \Big)^2 -\sum_{j=1}^N Q_j^2 \bigg\} 
    |\Psi^{(N)}(k_1,\cdots,k_N)|^2 \notag \\
& = \norm{\dGb(Q)\Psi}^2 - \norm{\dGb(Q^2)^{1/2}\Psi}^2,  \label{estxx2}
\end{align}
and, for $n\geq 2$, we have
\begin{align*}
\eqref{estxx1} 
& \leq  \sum_{N=n}^\infty  \int_{M^N} d\mu(k_1)\cdots d\mu(k_N)\,\Big( \sum_{j=1}^N Q(k_j)\Big)^n | \Psi^{(N)}(k_1,\cdots,k_N)|^2 \\
& = \norm{\dGb(Q)^{n/2}\Psi} ^2 < \infty.
\end{align*}
Hence, for all $\Phi\in\Fbz$ and $\Psi\in\dom(\dGb(Q)^{n/2})$, it holds that 
\begin{align*}
  |\inner{A^*(f_1)\cdots A^*(f_j)\Phi}{\Psi}|
 \leq \prod_{\ell=1}^j \norm{Q^{-1/2}f_\ell} \cdot \norm{\Phi}\cdot \norm{\dGb(Q)^{j/2}\Psi},\qquad j=1,2,\cdots,n.
\end{align*}
Since $\Fbz$ is a core for $A(f)$, by setting $j=1$ in the above inequality, we have $\Psi\in\dom(A(f_1))$.
Next, by setting $j=2$, one has $A(f_1)\Psi \in\dom(A(f_2))$.
Therefore $\Psi\in \dom(A(f_n)\cdots A(f_1))$ follows by induction.
We also have the bound
\begin{align*}
 \norm{A(f_n)\cdots A(f_1)\Psi} 
&= \mathop{\sup_{\Phi\in\Fbz}}_\mathrm{\norm{\Phi}=1}
  |\inner{A^*(f_1)\cdots A^*(f_n)\Phi}{\Psi}| \\
& \leq \prod_{\ell=1}^n \norm{Q^{-1/2}f_\ell} \cdot \norm{\dGb(Q)^{n/2}\Psi}.
\end{align*}
Thus we get \eqref{ak3-1}. The bound \eqref{ak3-2} follows from \eqref{estxx2}.
\end{proof}

\begin{lem}{\label{phi2<H0}}
Let $T$ be an injective self-adjoint operator and $g\in\dom(T^{-1/2})$.
Then $\dom(\dGb(T))\subset \dom(\PhiS(g)^2)$ and for all $\Psi\in\dom(\dGb(T))$,
\begin{align}
  \frac{1}{2} \norm{\PhiS(g)^2\Psi} 
 \leq \norm{T^{-1/2}g}^2 \norm{\dGb(T)\Psi} + \norm{g}^2\norm{\Psi} \label{phi2<H0x}
\end{align}
holds.
\end{lem}
\begin{proof}
In this proof, we write
\begin{align*}
 a:=A(g), ~~ a^*:=A^*(g), ~~ 
 c:=\norm{g}^2, ~~ d:=\norm{T^{-1/2}g}^2, ~~
 H_0 := \dGb(T)
\end{align*}
for short. 
We first assume $\Psi\in\Ffin(\dom(T))$.
By the triangle inequality,
\begin{align*}
 \norm{\PhiS(g)^2\Psi}^2 
 & = \frac{1}{4} \norm{(a^2 + a^*a + aa^* + a^{*2})\Psi}^2 \\
 & \leq \norm{a^2\Psi}^2 + \norm{a^*a\Psi}^2 + \norm{aa^*\Psi}^2 + \norm{a^{*2}\Psi}^2.
\end{align*}
By the CCRs and Lemma \ref{ak3}, we have
\begin{align*}
 \norm{a^2\Psi}^2 
&\leq d^2 \norm{H_0\Psi}^2, \\
 \norm{a^*a\Psi}^2 
& = \norm{a^2\Psi}^2 + c\norm{a\Psi}^2
  \leq d^2\norm{H_0\Psi}^2 + cd\norm{H_0^{1/2}\Psi}^2, \\
\norm{aa^*\Psi}^2 
& = \norm{(a^*a+c)\Psi}^2 
  = \norm{a^*a\Psi}^2 + 2c\norm{a\Psi}^2 + c^2\norm{\Psi}^2 \\
& \leq d^2\norm{H_0\Psi}^2+3cd\norm{H_0^{1/2}\Psi}^2 + c^2\norm{\Psi}^2, \\
\norm{(a^*)^2\Psi}^2 
 & = \norm{aa^*\Psi}^2 + c\norm{a^*\Psi}^2 \\
 & \leq  d^2\norm{H_0\Psi}^2+3cd\norm{H_0^{1/2}\Psi}^2 + c^2\norm{\Psi}^2 
   + c(d\norm{H_0^{1/2}\Psi}^2+c\norm{\Psi}^2) \\
 & =  d^2\norm{H_0\Psi}^2+4cd\norm{H_0^{1/2}\Psi}^2 + 2c^2\norm{\Psi}^2.
\end{align*}
Thus, we have
\begin{align*}
  \norm{\PhiS(g)^2\Psi}^2
& \leq 4d^2\norm{H_0\Psi}^2 + 8cd\norm{H_0^{1/2}\Psi}^2 + 3c^2\norm{\Psi}^2 \\
& \leq \norm{(2dH_0+2c)\Psi}^2,
\end{align*}
and the bound \eqref{phi2<H0x} holds for all $\Psi\in\Ffin(\dom(T))$.
By a limiting argument, the lemma follows.
\end{proof}

\section{On the Domains of $T^{-3/2}$ and $S^{-3/2}$}

It is shown that $\dom(T^{p})=\dom(S^{p})$ for all $|p|\leq 1$ and $p=2$ in Lemma \ref{bounded}.
Here, we show the equality for $p=-3/2$ under the infrared regularity condition.

\begin{lem}\label{domain-3/2}
Suppose (B1)--(B5).
We further suppose that $g_n\in\dom(T^{-1})$ for all $n\in\NN$, and that
\[
\sum_{n=1}^{\infty}|\lambda_n|\cdot\|g_n\|\cdot\|T^{-1}g_n\|<\infty.
\]
Then $\dom(T^{-3/2})=\dom(S^{-3/2})$ holds, where $S$ is defined in \eqref{defS}.
\end{lem}

\begin{proof}
Let $A:=\sum_{n=1}^\infty\lambda_n\ket{g_n}\bra{T^{-1}g_n}$.
By assumption, $A$ is of trace class.
We first show that $1+A$ is bijective.
Let $u\in\ker(1+A)$ be arbitrary.
Since the range of $A$ is contained in $\dom(T^{-1/2})$, the equality $(1+A)u=0$ implies that $u\in\dom(T^{-1/2})$.
Hence it follows that
\begin{align*}
0 &= T^{-1/2}(1+A)u = T^{-1/2}u + \sum_{n=1}^\infty\lambda_n\langle T^{-1/2}g_n,T^{-1/2}u\rangle T^{-1/2}g_n\\
&=\left(1+\sum_{n=1}^\infty\lambda_n\ket{T^{-1/2}g_n}\bra{T^{-1/2}g_n}\right)T^{-1/2}u.
\end{align*}
This together with the condition (B5) implies that $u=0$, and thus $1+A$ is injective.
The Fredholm alternative now tells us that $1+A$ is bijective. 

Set
\[
\sD :=\dom(T^2)\cap\dom(T^{-1/2})=\dom(S^2)\cap\dom(S^{-1/2}),
\]
and take an arbitrary $u\in\sD $.
Then we have
\begin{align*}
&S^{3/2}u=S^{-1/2}S^2u=S^{-1/2}\left(T^2+\sum_{n=1}^\infty\lambda_n\ket{T^{1/2}g_n}\bra{T^{1/2}g_n}\right)u\\
&=S^{-1/2}T^{1/2}\left(1+\sum_{n=1}^\infty\lambda_n\ket{g_n}\bra{T^{-1}g_n}\right)T^{3/2}u
=\overline{S^{-1/2}T^{1/2}}(1+A)T^{3/2}u.
\end{align*}
Note that $\overline{T^{1/2}S^{-1/2}}$ is bijective with inverse $\overline{S^{1/2}T^{-1/2}}$ by Lemma \ref{domain}.
Since $\sD $ is a core of both $S^{3/2}$ and $T^{3/2}$, 
we get the operator equality
\[
S^{3/2} = \overline{S^{-1/2}T^{1/2}}(1+A)T^{3/2}.
\]
By taking the conjugation of both sides, we obtain the operator equality
\begin{equation}\label{douglas}
S^{3/2} = T^{3/2}(1+A^*)\overline{T^{1/2}S^{-1/2}},
\end{equation}
which in particular implies that the range of $S^{3/2}$ is contained in the range of $T^{3/2}$,
and thus $\dom(S^{-3/2})\subset\dom(T^{-3/2})$ holds.
On the other hand, it follows from \eqref{douglas} that we have the operator equality
\[
S^{3/2}\overline{S^{1/2}T^{-1/2}}(1+A^*)^{-1} = T^{3/2},
\]
whence $\dom(S^{-3/2})\supset\dom(T^{-3/2})$ follows.
This completes the proof.
\end{proof}

\section*{Data Availability}
No data were used to support this study.

\vspace*{10pt}
\noindent\textbf{Acknowledgments:}
We thank Shinnosuke Izumi for pointing out several missprints.
This work was supported by JSPS KAKENHI (Grant Number JP16K17612 and JP20K03628).



{\small
\begin{thebibliography}{00}

\bibitem{Ar81}
 {A. Arai},
 {On a model of a harmonic oscillator coupled to a quantized, massless, scalar field. I},
 \textit{J. Math. Phys.}, \textbf{22}, 2539--2548, (1981)

\bibitem{A83} 
 {A. Arai},
 {Rigorous theory of spectra and radiation for a model in quantum electrodynamics},
 \textit{J.~Math.~Phys.}, {\bf 24}, 1896--1910, (1983)
 
\bibitem{A83b} 
{A. Arai},
 {A note on scattering theory in non-relativistic quantum electrodynamics},
 \textit{J. Phys. A: Math. Gen.}, {\bf16}, 49--70, (1983)
 
 \bibitem{A00}
 {A. Arai},
 \textit{Fock spaces and Quantum fields},
 (in Japanese) Nippon-hyoronsha, Tokyo, (2000)
 
\bibitem{A18}
  {A. Arai},
  \textit{Analysis of Fock spaces and Mathematical theory of quantum fields},
  {World Scientific}, (2018)

\bibitem{AF}
 {K. Asahara and D. Funakawa},
 {Spectral analysis of an abstract pair interaction model}, 
 to appear in \textit{Hokkaido Math. J.}, arXiv:1807.08408v1

\bibitem{Be}
  {E. A. Berezin},
  \textit{The Method of Second Quantization},
  {Academic Press}, (1966)

\bibitem{De03}
   {J. Derezi\'{n}ski},
   {Van Hove Hamiltonians -- exactly solvable models of the infrared and ultraviolet problem},
   \textit{Ann. Henri Poincare}, \textbf{4}, 713--738, (2003)

\bibitem{De17}
  {J. Derezi\'{n}ski},
  {Bosonic quadratic Hamiltonians},
  \textit{J. Math. Phys.}, \textbf{58}, 121101, (2017)

\bibitem{DG13}
  {J. Derezi\'{n}ski and C. G\'{e}rard},
  \textit{Mathematics of Quantization and Quantum Fields},
  {Cambridge University Press}, (2013)

\bibitem{EB}
  {F. Hiroshima, I. Sasaki, H. Spohn and A. Suzuki}, 
  \textit{Enhanced Binding in Quantum Field Theory}, 
  {COE Lecture Note Vol. 38}, IMI, Kyushu University, (2012)

\bibitem{GS}
  {P. Grech and R. Seiringer},
  \textit{The Excitation Spectrum for Weakly Interacting Bosons in a Trap},
  {Comm. Math. Phys.}, \textbf{322}, 559--591, (2013)

\bibitem{KM}
  {A. Klein and B. H. McCormick},
  {Meson Pair Theory}, 
  \textit{Phys. Rev.}, \textbf{98}, 1428--1445, (1955)

\bibitem{ms05}
  {T. Miyao and I. Sasaki},
  {Stability of discrete ground state}, 
  \textit{Hokkaido Math. J.}, \textbf{34}, 689-717, (2005)

\bibitem{TNS}
 {P. T. Nam, M. Napi\'{o}rkowskia and J. P. Solovej},
 {Diagonalization of bosonic quadratic Hamiltonians by Bogoliubov transformations},
 \textit{J. Funct. Anal.}, \textbf{270}, 4340--4368, (2016)

\bibitem{Rui78}
  {S. N. M. Ruijsenaars},
  {On Bogoliubov Transformations.\ II. The General Case}, 
  \textit{Ann.\ Phys.} \textbf{116}, {105--134}, (1978)

\bibitem{Sc12}
  {K. Schm\"udgen},
  \textit{Unbounded self-adjoint operators on Hilbert space}, 
  volume 265 of Graduate Texts in Mathematics, Springer, (2012)

\bibitem{Si15}
 {B. Simon},
 \textit{A Comprehensive Course in Analysis, Part 4: Operator Theory}, 
 American Mathematical Society, Providence, RI, (2015)
\end{thebibliography}
}
\end{document}